\documentclass[11pt]{article}
\usepackage{amsmath}
\usepackage{amsthm}
\usepackage{amssymb}
\usepackage{algorithm}
\usepackage{subfig}
\usepackage{color}
\usepackage[english]{babel}
\usepackage{graphicx}
\usepackage{wrapfig,epsfig}
\usepackage{epstopdf}
\usepackage{url}
\usepackage{graphicx}
\usepackage{color}
\usepackage{epstopdf}
\usepackage{algpseudocode}
\usepackage{scrextend}
\usepackage[T1]{fontenc}
\usepackage{bbm}
\usepackage{comment}

\usepackage{tikz}
\usepackage{hyperref}  
\hypersetup{colorlinks=true,citecolor=blue,linkcolor=blue} 
\usetikzlibrary{arrows}
\usepackage[margin=1in]{geometry}
\graphicspath{{./figs/}}

\usepackage{etoolbox}

\makeatletter

\newcommand{\define}[4][ignore]{%
  \ifstrequal{#1}{ignore}{}{
  \@namedef{thmtitle@#3}{#1}}%
  \@namedef{thm@#3}{#4}%
  \@namedef{thmtypen@#3}{lemma}%
  \newtheorem{thmtype@#3}[theorem]{#2}%
  \newtheorem*{thmtypealt@#3}{#2~\ref{#3}}%
}

\newcommand{\state}[1]{%
  \@namedef{curthm}{#1}
  \@ifundefined{thmtitle@#1}{
  \begin{thmtype@#1}
    }{
  \begin{thmtype@#1}[\@nameuse{thmtitle@#1}]
  }
    \label{#1}
    \@nameuse{thm@#1}
  \end{thmtype@#1}
  \@ifundefined{thmdone@#1}{
  \@namedef{thmdone@#1}{stated}%
  }{}
}

\newcommand{\restate}[1]{%
  \@namedef{curthm}{#1}
  \@ifundefined{thmtitle@#1}{
    \begin{thmtypealt@#1}
    }{
  \begin{thmtypealt@#1}[\@nameuse{thmtitle@#1}]
  }
    \@nameuse{thm@#1}
  \end{thmtypealt@#1}
  \@ifundefined{thmdone@#1}{
  \@namedef{thmdone@#1}{stated}%
  }{}
}

\newcommand{\thmlabel}[1]{
  \@ifundefined{thmdone@\@nameuse{curthm}}{\label{#1}
    }{\tag*{\eqref{#1}}}
}
\makeatother

\newtheorem{theorem}{Theorem}
\newtheorem{conjecture}[theorem]{Conjecture}
\newtheorem{definition}[theorem]{Definition}
\newtheorem{lemma}[theorem]{Lemma}
\newtheorem{proposition}[theorem]{Proposition}
\newtheorem{corollary}[theorem]{Corollary}
\newtheorem{claim}[theorem]{Claim}
\newtheorem{fact}[theorem]{Fact}

\newtheorem{remk}[theorem]{Remark}
\newtheorem{exmp}[theorem]{Example}

\newenvironment{example}{\begin{exmp}
\begin{normalfont}}{\end{normalfont}
\end{exmp}}

\newenvironment{remark}{\begin{remk}
\begin{normalfont}}{\end{normalfont}
\end{remk}}

\def\FullBox{\hbox{\vrule width 8pt height 8pt depth 0pt}}

\def\qed{\ifmmode\qquad\FullBox\else{\unskip\nobreak\hfil
\penalty50\hskip1em\null\nobreak\hfil\FullBox
\parfillskip=0pt\finalhyphendemerits=0\endgraf}\fi}

\def\qedsketch{\ifmmode\Box\else{\unskip\nobreak\hfil
\penalty50\hskip1em\null\nobreak\hfil$\Box$
\parfillskip=0pt\finalhyphendemerits=0\endgraf}\fi}

\newcommand{\eqdef}{\mathbin{\stackrel{\rm def}{=}}}
\newcommand{\R}{{\mathbb R}}

\newcommand{\poly}{{\mathrm{poly}}}

\newcommand{\eps}{\varepsilon}

\newcommand{\inner}[1]{\left\langle#1\right\rangle}

\newcommand{\norm}[1]{\|#1\|}
\newcommand{\abs}[1]{|#1|}
\newcommand{\Ot}{\widetilde{O}}
\DeclareMathOperator*{\E}{{\bf E}}
\DeclareMathOperator*{\V}{{\bf V}}
\DeclareMathOperator*{\argmin}{arg\,min}
\DeclareMathOperator{\nnz}{nnz}
\DeclareMathOperator{\hard}{hard}
\DeclareMathOperator{\rank}{rank}
\DeclareMathOperator{\OPT}{OPT}
\DeclareMathOperator{\supp}{supp}

\renewcommand{\hat}{\widehat}
\newcommand{\wt}{\widetilde}
\newcommand{\ov}{\overline}
\makeatletter
\newcommand*{\RN}[1]{\expandafter\@slowromancap\romannumeral #1@}
\makeatother

\title{Fast Regression with an $\ell_\infty$ Guarantee\thanks{A preliminary version of this paper appears in Proceedings of the 44th International Colloquium on Automata, Languages, and Programming  (ICALP 2017).}}

\author{
  Eric Price\\
  \texttt{zhaos@utexas.edu}\\
  UT-Austin
  \and
  Zhao Song\\
  \texttt{zhaos@utexas.edu}\\
  UT-Austin
  \and
  David P. Woodruff\\
  \texttt{dpwoodru@us.ibm.com}\\
  IBM Almaden
}

\begin{document}

\begin{titlepage}
  \maketitle

\begin{abstract}
  Sketching has emerged as a powerful technique for speeding up
  problems in numerical linear algebra, such as regression.
  In the overconstrained
  regression problem, one is given
  an $n \times d$ matrix $A$, with $n \gg d$, as well as an $n \times
  1$ vector $b$, and one wants to find a vector $\hat{x}$ so as to
  minimize the residual error $\|Ax-b\|_2$. Using the sketch and solve
  paradigm, one first computes $S \cdot A$ and $S \cdot b$ for a
  randomly chosen matrix $S$, then outputs
  $x' = (SA)^{\dagger} Sb$ so as to minimize $\norm{SAx' - Sb}_2$.

  The sketch-and-solve paradigm gives a bound on
  $\norm{x'-x^*}_2$ when $A$ is well-conditioned.  Our main
  result is that, when $S$ is the subsampled randomized
  Fourier/Hadamard transform, the error $x' - x^*$ behaves as if
  it lies in a ``random'' direction within this bound: for any fixed
  direction $a\in \R^d$, we have with $1 - d^{-c}$ probability that
  \begin{align}\label{eq:abstract}
    \inner{a, x'-x^*} \lesssim \frac{\norm{a}_2\|x'-x^*\|_2}{d^{\frac{1}{2}-\gamma}},
  \end{align}
  where $c, \gamma > 0$ are arbitrary constants.
  This implies $\|x'-x^*\|_{\infty}$
  is a factor $d^{\frac{1}{2}-\gamma}$
  smaller than $\|x'-x^*\|_2$.  It also gives 
  a better bound on the generalization of $x'$ to new examples:
  if rows of $A$ correspond to examples and columns to features, then
  our result gives a better bound for the error introduced by
  sketch-and-solve when classifying fresh examples.
  We show that not all oblivious subspace embeddings
  $S$ satisfy these properties. In particular, we give counterexamples
  showing that matrices based on
  Count-Sketch or leverage score sampling do not satisfy these properties. 

  We also provide lower bounds, both on how small $\|x'-x^*\|_2$ can
  be, 
  and for our new guarantee~\eqref{eq:abstract},
  showing that the subsampled randomized
  Fourier/Hadamard transform is nearly optimal. Our lower bound 
  on $\|x'-x^*\|_2$ shows
  that there is an $O(1/\eps)$ separation in the dimension
  of the optimal oblivious subspace embedding required for outputting an $x'$
  for which $\|x'-x^*\|_2 \leq \epsilon \|Ax^*-b\|_2 \cdot \|A^{\dagger}\|_2$,
  compared to the dimension of the optimal oblivious subspace embedding
  required for outputting an $x'$ for which
  $\|Ax'-b\|_2 \leq (1+\epsilon)\|Ax^*-b\|_2$, that is, the former
  problem requires dimension $\Omega(d/\epsilon^2)$ while the latter
  problem can be solved with dimension $O(d/\epsilon)$. This explains
  the reason known upper bounds on the
  dimensions of these two variants of regression
  have differed in prior work. 

\end{abstract}

  \thispagestyle{empty}
\end{titlepage}

\section{Introduction}
Oblivious subspace embeddings (OSEs) were introduced by
Sarlos~\cite{sar06} to solve linear algebra problems more quickly than
traditional methods.  An OSE is a distribution of matrices $S \in
\R^{m \times n}$ with $m \ll n$ such that, for any $d$-dimensional
subspace $U \subset \R^n$, with ``high'' probability $S$ preserves the
norm of every vector in the subspace.  OSEs are a generalization of
the classic Johnson-Lindenstrauss lemma from vectors to subspaces.
Formally, we require that with probability $1-\delta$, 
\[
\norm{Sx}_2 = (1 \pm \eps) \norm{x}_2
\]
simultaneously for all $x \in U$, that is,
$(1-\eps) \norm{x}_2 \leq \norm{Sx}_2 \leq (1+\eps) \norm{x}_2$. 

A major application of OSEs is to regression.  The regression problem
is, given $b \in \R^n$ and $A \in \R^{n \times d}$ for $n \geq d$, to
solve for
\begin{align}
  x^* = \argmin_{x \in \mathbb{R}^d} \norm{Ax-b}_2\label{eq:regression}
\end{align}
Because $A$ is a ``tall'' matrix with more rows than columns, the
system is overdetermined and there is likely no solution to $Ax = b$,
but regression will find the closest point to $b$ in the space spanned
by $A$.  The classic answer to regression is to use the Moore-Penrose
pseudoinverse: $x^* = A^\dagger b$ where
\[
A^\dagger = (A^\top A)^{-1}A^\top
\]
is the ``pseudoinverse'' of $A$ (assuming $A$ has full column rank,
which we will typically do for simplicity).  This classic solution
takes $O(nd^{\omega - 1} + d^{\omega})$ time, where $\omega < 2.373$ is
the matrix multiplication constant~\cite{cw90,w12,g4a}: $nd^{\omega -
  1}$ time to compute $A^\top A$ and $d^{\omega}$ time to compute the inverse.

OSEs speed up the process by replacing~\eqref{eq:regression} with
\[
x' = \argmin_{x\in \R^d} \norm{SAx-Sb}_2
\]
for an OSE $S$ on $d+1$-dimensional spaces.  This replaces the $n
\times d$ regression problem with an $m \times d$ problem, which can be
solved more quickly since $m \ll n$.  Because $Ax - b$ lies in the
$d+1$-dimensional space spanned by $b$ and the columns of $A$, with
high probability $S$ preserves the norm of $SAx - Sb$ to $1 \pm \eps$
for all $x$.  Thus,
\[
\norm{Ax'-b}_2 \leq \frac{1+\eps}{1-\eps} \norm{Ax^*-b}_2.
\]
That is, $S$ produces a solution $x'$ which preserves the {\it cost}
of the regression problem.  The running time for this method depends
on (1) the reduced dimension $m$ and (2) the time it takes to multiply
$S$ by $A$.  We can compute these for ``standard'' OSE types:

\begin{itemize}
\item If $S$ has i.i.d. Gaussian entries, then $m = O(d/\eps^2)$ is
  sufficient (and in fact, $m \geq d/\epsilon^2$ is required \cite{nn14}).
  However, computing $SA$ takes $O(mnd) =
  O(nd^2/\eps^2)$ time, which is worse than solving the original
  regression problem (one can speed this up using fast matrix multiplication,
  though it is still worse than solving the original problem). 
\item If $S$ is a subsampled randomized Hadamard transform (SRHT)
  matrix with random sign
  flips~(see Theorem 2.4 in \cite{woo14} for a survey, and also see 
\cite{CNW16} which gives a recent improvement)
then $m$
  increases to $\Ot(d/\eps^2 \cdot \log n)$, where $\Ot(f) = f
  \poly(\log(f))$.  But now, we can compute $SA$ using the fast
  Hadamard transform in $O(nd\log n)$ time.  This makes the overall
  regression problem take $O(nd\log n + d^\omega/\eps^2)$ time.
\item If $S$ is a random sparse matrix with random signs (the
  ``Count-Sketch'' matrix), then $m = d^{1 + \gamma}/\eps^2$ suffices
  for $\gamma > 0$ a decreasing function of the 
  sparsity~\cite{CW13,MM13,NN13,BDN15,c16}.
  (The definition of a Count-Sketch matrix
  is, for any $s\geq 1$, $S_{i,j}\in \{0, -1/\sqrt{s}, 1/\sqrt{s} \}$, 
  $\forall i\in [m], j\in  [n]$ and the column sparsity of matrix $S$ is $s$.
  Independently in
  each column $s$ positions are chosen uniformly at random without
  replacement, and each chosen position is set to $-1/\sqrt{s}$ with
  probability $1/2$, and $+1/\sqrt{s}$ with probability $1/2$.)
  Sparse OSEs can benefit from the sparsity of $A$, allowing for a running time of
  $\Ot(\nnz(A)) + \Ot(d^\omega/\eps^2)$, 
  where $\nnz(A)$ denotes the number of non-zeros in $A$.
\end{itemize}
When $n$ is large, the latter two algorithms are substantially faster
than the na\"ive $nd^{\omega-1}$ method.



\subsection{Our Contributions}
Despite the success of using subspace embeddings to speed up regression, 
often what practitioners are interested is not in preserving the cost
of the regression problem, but rather in the {\it generalization} or 
{\it prediction} error 
provided by the vector $x'$. Ideally, we would like for any future (unseen) example
$a \in \mathbb{R}^d$, that 
$\langle a, x' \rangle \approx \langle a, x^* \rangle$ with
high probability.

Ultimately one may want to use $x'$ to do classification,
such as regularized least squares classification (RLSC) \cite{ryp03}, 
which has been found in
cases to do as well as support vector machines but is much simpler \cite{zp04}.
In this application, given a training 
set of examples with multiple (non-binary)  
labels identified with the rows of an $n \times d$ 
matrix $A$, one creates an $n \times r$ matrix $B$, each column indicating the presence
or absence of one of the $r$ possible labels in each example. One then solves 
the multiple response regression problem $\min_X \|AX-B\|_F$, and uses $X$ to classify
future examples. A commonly used method is for a future example $a$, to compute
$\langle a, x_1 \rangle, \ldots, \langle a, x_r \rangle$, where $x_1, \ldots, x_r$
are the columns of $X$. One then chooses the label $i$ for which $\langle a, x_i \rangle$ is maximum. 

For this to work, we would like the inner products $\langle a,
x'_1 \rangle, \ldots, \langle a, x'_r \rangle$ to be close to $\langle
a, x_1^* \rangle, \ldots$, $\langle a, x_r^* \rangle$, where $X'$ is the
solution to $\min_X \|SAX-SB\|_F$ and $X^*$ is the solution to $\min_X
\|AX-B\|_F$. For any $O(1)$-accurate OSE on $d+r$ dimensional spaces \cite{sar06}, which also satisfies so-called 
approximate matrix multiplication with error $\eps' = \eps/\sqrt(d+r)$, we get that
\begin{align}\label{eq:ell2}
  \norm{x' - x^*}_2 \leq O(\eps) \cdot \norm{Ax^* - b}_2 \cdot \norm{A^\dagger}_2
\end{align}
where $\norm{A^\dagger}$ is the spectral norm of $A^\dagger$, which equals the
reciprocal of the smallest singular value of $A$. To obtain a generalization
error bound for an unseen example $a$, one has
\begin{eqnarray}\label{eqn:old}
|\langle a, x^* \rangle - \langle a, x' \rangle | = |\langle a, x^*-x' \rangle| \leq \|x^*-x'\|_2 \|a\|_2 = O(\eps) \|a\|_2 \norm{Ax^*-b}_2 \norm{A^\dagger}_2,
\end{eqnarray}
which could be tight if given only the guarantee in \eqref{eq:ell2}.
However, if the difference vector $x' - x^*$ were distributed in a uniformly
random direction subject to~\eqref{eq:ell2}, then one would expect an
$\Ot(\sqrt{d})$ factor improvement in the bound.  This is what our main
theorem shows:
\begin{theorem}[Main Theorem, informal]\label{thm:main}
  Suppose $n \leq \poly(d)$ and matrix $A\in \R^{n\times d}$ and vector $b\in \mathbb{R}^{n}$ are given.  Let $S\in \R^{m\times n}$ be a subsampled randomized
  Hadamard transform matrix with $m = d^{1+\gamma}/\eps^2$ rows for an
  arbitrarily small constant $\gamma > 0$. For $x' = \argmin_{x\in \R^d}
  \norm{SAx-Sb}_2$ and $x^* = \argmin_{x\in \R^d} \norm{Ax-b}_2$, and any fixed
  $a \in \R^d$,
\begin{eqnarray}\label{eq:general}
|\langle a, x^* \rangle - \langle a, x' \rangle |
\leq \frac{\eps}{\sqrt{d}} \|a\|_2 \norm{Ax^*-b}_2 \norm{A^\dagger}_2.
\end{eqnarray}
with probability $1-1/d^C$ for an arbitrarily large constant $C > 0$.
This implies that
\begin{eqnarray}\label{eq:ellinf}
\|  x^* - x' \|_{\infty}
\leq \frac{\eps}{\sqrt{d}}  \norm{Ax^*-b}_2 \norm{A^\dagger}_2.
\end{eqnarray}
with $1 - 1/d^{C-1}$ probability.

If $n > \poly(d)$, then by first composing $S$ with a Count-Sketch OSE with $\poly(d)$
rows, one can achieve the same guarantee. 
\end{theorem}
(Here $\gamma$ is a constant going to zero as $n$ increases;
see Theorem \ref{thm:eps} for a formal statement of Theorem \ref{thm:main}.)

Notice that Theorem \ref{thm:main} is considerably stronger than
that of (\ref{eqn:old}) provided by existing guarantees. Indeed, in order
to achieve the guarantee (\ref{eq:ellinf}) 
in Theorem \ref{thm:main}, one would need
to set $\eps' = \eps/\sqrt{d}$ in existing OSEs, resulting in $\Omega(d^2/\epsilon^2)$
rows. In contrast, we achieve only $d^{1+\gamma}/\epsilon^2$ rows. 
We can improve the bound in Theorem \ref{thm:main} to $m = d/\eps^2$ 
if $S$ is a matrix
of i.i.d. Gaussians; however, as noted, computing $S \cdot A$ is slower in
this case. 


Note that Theorem \ref{thm:main} also {\it makes no distributional
  assumptions} on the data, and thus the data could be heavy-tailed or
even adversarially corrupted. This implies that our bound is still useful
when the rows of $A$ are not sampled independently from a distribution
with bounded variance. 

The $\ell_\infty$ bound~\eqref{eq:ellinf} of Theorem~\ref{thm:main} is
achieved by applying~\eqref{eq:general} to the standard basis vectors
$a = e_i$ for each $i \in [d]$ and applying a union bound.  This
$\ell_{\infty}$ guarantee often has a more natural interpretation than
the $\ell_2$ guarantee---if we think of the regression as attributing
the observable as a sum of various factors,~\eqref{eq:ellinf} says
that the contribution of each factor is estimated well.  One may also
see our contribution as giving a way for estimating the pseudoinverse
$A^\dagger$ \emph{entrywise}. Namely, we get that $ (SA)^\dagger S
\approx A^\dagger $ in the sense that each entry is within additive $O(\eps
\sqrt{\frac{\log d}{d}} \norm{A^\dagger}_2)$.  There is a lot of work
on computing entries of inverses of a matrix, see, e.g.,
\cite{ADLRRU12,LiAKD08}.

Another benefit of the $\ell_\infty$ guarantee is when the regression
vector $x^*$ is expected to be
$k$-\emph{sparse}~(e.g. \cite{leekasso}).  In such cases, thresholding
to the top $k$ entries will yield an $\ell_2$ guarantee a factor
$\sqrt{\frac{k}{d}}$ better than~\eqref{eq:ell2}.

%
One could ask if Theorem \ref{thm:main} also holds for sparse OSEs, such
as the Count-Sketch. Surprisingly, we show that one cannot achieve the 
generalization error guarantee in Theorem \ref{thm:main} 
with high probability, say, $1-1/d$, 
using such embeddings, despite the fact that such
embeddings do approximate the cost of the regression problem up to a 
$1+\epsilon$ factor with high probability. This shows that the generalization
error guarantee is achieved by some subspace embeddings but not all. 

\begin{theorem}[Not all subspace embeddings give the $\ell_{\infty}$ guarantee; informal version of Theorem \ref{thm:count_sketch_not_infty}] 
  The Count-Sketch matrix with $d^{1.5}$ rows and sparsity
  $d^{.25}$---which is an OSE with exponentially small failure probability---with
  constant probability will have a result $x'$ that does not satisfy the
  $\ell_{\infty}$ guarantee (\ref{eq:ellinf}).
\label{thm:negative}
\end{theorem}
We can show that Theorem \ref{thm:main} holds for $S$ based
on the Count-Sketch OSE $T$ with $d^{O(C)}/\epsilon^2$ rows with $1-1/d^C$ probability. 
We can thus compose the Count-Sketch OSE with the SRHT matrix
and obtain an $O(\nnz(A)) + \poly(d/\epsilon)$ time algorithm to compute $S \cdot T A$
achieving (\ref{eq:ellinf}). We can also compute $R \cdot S \cdot T \cdot A$, where
$R$ is a matrix of Gaussians, which is more efficient now that $S T A$ only has
$d^{1+\gamma}/\epsilon^2$ rows; this will reduce the number of rows to $d/\epsilon^2$. 

Another common method of dimensionality reduction for linear
regression is \emph{leverage score
  sampling}~\cite{DMMW12,LMP13,PKB14,CMM15}, which subsamples the rows
of $A$ by choosing each row with probability proportional to its 
``leverage scores''.  With $O(d \log(d/\delta)/\eps^2)$ rows taken,
the result $x'$ will satisfy the $\ell_2$ bound~\eqref{eq:ell2} with
probability $1-\delta$.
However, it does not give a good $\ell_\infty$ bound:

\begin{theorem}[Leverage score sampling does not give the $\ell_{\infty}$ guarantee; informal version of Theorem \ref{thm:leverage_score_not_infty}]
  Leverage score sampling with $d^{1.5}$ rows---which satisfies the
  $\ell_2$ bound with exponentially small failure probability---with constant
  probability will have a result $x'$ that does not satisfy the
  $\ell_{\infty}$ guarantee (\ref{eq:ellinf}).
\end{theorem}

%
%

Finally, we show that the $d^{1+\gamma}/\eps^2$ rows that SRHT
matrices use is roughly optimal:
\begin{theorem}[Lower bounds for $\ell_2$ and $\ell_\infty$
  guarantees; informal versions of of Theorem \ref{thm:l2_lower_bound}
  and Corollary \ref{cor:linf_lower_bound}]
  Any sketching matrix distribution over $m \times n$ matrices that
  satisfies either the $\ell_{2}$ guarantee~\eqref{eq:ell2} or the
  $\ell_\infty$ guarantee~\eqref{eq:ellinf} must have $m \gtrsim
  \min(n,d/\eps^2)$.
\end{theorem}

Notice that our result shows the necessity of the $1/\eps$ separation
between the results originally defined in Equation (3) and (4) of
Theorem 12 of \cite{sar06}.  If we want to output some vector $x'$
such that $\| Ax'-b \|_2 \leq (1+\eps) \| A x^* -b\|_2$, then it is
known that $m=\Theta(d/\eps)$ is necessary and sufficient.  However,
if we want to output a vector $x'$ such that $\| x'-x^*\|_2 \leq \eps
\| A x^* -b \|_2 \cdot \|A^\dagger \|_2$, then we show that $m =
\Theta(d/\eps^2)$ is necessary and sufficient.

\subsubsection{Comparison to Gradient Descent}
While this work is primarily about sketching methods, one could instead apply 
iterative methods such as gradient descent, after appropriately
preconditioning the matrix, see, e.g., \cite{amt10,zf13, CW13}. 
That is, one can use an OSE with constant
$\eps$ to construct a preconditioner for $A$ 
and then run conjugate gradient using the
preconditioner. This gives an overall dependence of $\log(1/\epsilon)$. 

The main drawback of this approach is
that one loses
the ability to save on
storage space or number of passes when $A$ appears in a stream, or to save on
communication or rounds when $A$ is distributed. Given increasingly
large data sets, such scenarios are now quite common, see, e.g., \cite{CW09}
for regression algorithms in the data stream model. 
In situations where the entries of $A$ appear
sequentially, for example, a row at a time, one does not need to store
the full $n \times d$ matrix $A$ but only the $m \times d$ matrix
$SA$.

Also, iterative methods can be less efficient
when solving multiple response regression, where one wants to
minimize $\norm{AX - B}$ for a $d \times t$ matrix $X$ and an $n \times t$
matrix $B$. This is the case when $\eps$ is constant and $t$ is large, 
which can occur in some applications (though there are also other
applications for which $\eps$ is very small). For 
example, conjugate gradient with a preconditioner will take
$\Ot(ndt)$ time while using an OSE directly will take only $\Ot(nd +
d^2t)$ time (since one effectively replaces $n$ with $O~(d)$ after computing $S \cdot A$), separating $t$ from $d$. Multiple response regression, arises,
for example, in the RLSC application above. 

\subsubsection{Proof Techniques}
          \noindent {\bf Theorem \ref{thm:main}.}
          As noted in Theorem~\ref{thm:negative}, there are some OSEs for which
our generalization error bound does not hold. This hints that our
analysis is non-standard and cannot use generic properties of OSEs as
a black box. Indeed, in our analysis, we have to consider matrix
products of the form $S^\top S (UU^\top S^\top S)^k$ for our random
sketching matrix $S$ and a fixed matrix $U$, where $k$ is a positive
integer. We stress that it is the {\it same matrix} $S$ appearing
multiple times in this expression, which considerably complicates the
analysis, and does not allow us to appeal to standard results on
approximate matrix product (see, e.g., \cite{woo14} for a survey).
The key idea is
to recursively reduce $S^\top S (UU^\top S^\top S)^k$ using a property
of $S$. We use properties that only hold for specifics OSEs $S$:
first, that each column of $S$ is unit vector; and
second, that for all pairs $(i,j)$ and $i\neq j$, the inner product
between $S_i$ and $S_j$ is at most $\frac{\sqrt{\log n}}{\sqrt{m}}$
with probability $1-1/\poly(n)$.
\\\\
\noindent {\bf Theorems \ref{thm:count_sketch_not_infty} and \ref{thm:leverage_score_not_infty}.}
To show that Count-Sketch does not give the $\ell_{\infty}$ guarantee,
we construct a matrix $A$ and vector $b$ as in Figure
\ref{fig:A_b_count_sketch_not_linf}, which has optimal solution $x^*$
with all coordinates $1/\sqrt{d}$.  We then show, for our setting of
parameters, that there likely exists an index $j \in [d]$ satisfying
the following property: the $j$th column of $S$ has disjoint support
from the $k$th column of $S$ for all $k \in [d+\alpha]\setminus \{j\}$
except for a single $k > d$, for which $S_j$ and $S_k$ share exactly
one common entry in their support. In such cases we can compute $x'_j$
explicitly, getting $\abs{x'_j - x^*_j} = \frac{1}{s\sqrt{\alpha}}$.
By choosing suitable parameters in our construction, this gives that
$\norm{x'-x^*}_\infty \gg \frac{1}{\sqrt{d}}$. 
The lower bound for leverage score sampling follows a similar
construction.
\\\\
\noindent {\bf Theorem \ref{thm:l2_lower_bound} and Corollary \ref{cor:linf_lower_bound}.}
The lower bound proof for the $\ell_{2}$ guarantee uses Yao's minimax
principle. We are allowed to fix an $m\times n$ sketching matrix $S$
and design a distribution over $[A ~ b]$. We first write the sketching matrix $S = U \Sigma V^\top$ in its singular value decomposition (SVD). We choose the $d+1$ columns of the adjoined matrix $[A, b]$ to be random orthonormal vectors. Consider an $n \times n$ orthonormal matrix $R$ which contains the columns of $V$ as its first $m$ columns, and is completed on its remaining $n-m$ columns to an arbitrary orthonormal basis. Then $S \cdot [A, b] = V^\top R R^\top \cdot [A,b] = [U \Sigma I_m, 0] \cdot [R^\top A, R^\top b]$. Notice that $[R^\top A, R^\top b]$ is equal in distribution to $[A, b]$, since $R$ is fixed and $[A, b]$ is a random matrix with $d+1$ orthonormal columns. Therefore, $S \cdot [A, b]$ is equal in distribution to $[U \Sigma G, U \Sigma h]$ where $[G, h]$ corresponds to the first $m$ rows of an $n \times (d+1)$ uniformly random matrix with orthonormal columns. 

A key idea is that if $n = \Omega(\max(m,d)^2)$, then by a result of Jiang \cite{J06}, any $m \times (d+1)$ submatrix of a random $n \times n$ orthonormal matrix has $o(1)$ total variation distance to a $d \times d$ matrix of i.i.d. $N(0,1/n)$ random variables, and so any events that would have occurred had $G$ and $h$ been independent i.i.d. Gaussians, occur with the same probability for our distribution up to an $1-o(1)$ factor, so we can assume $G$ and $h$ are independent i.i.d. Gaussians in the analysis. 

The optimal solution $x'$ in the sketch space equals $(SA)^{\dagger} Sb$, and by using that $SA$ has the form $U \Sigma G$, one can manipulate $\|(SA)^{\dagger} Sb\|$ to be of the form $\|\tilde{\Sigma}^{\dagger} (\Sigma R)^{\dagger} \Sigma h\|_2$, where the SVD of $G$ is $R \tilde{\Sigma} T$. We can upper bound $\|\tilde{\Sigma}\|_2$ by $\sqrt{r/n}$, since it is just the maximum singular value of a Gaussian matrix, where $r$ is the rank of $S$, which allows us to lower bound $\|\tilde{\Sigma}^{\dagger} (\Sigma R)^{\dagger} \Sigma h\|_2$ by $\sqrt{n/r}\|(\Sigma R)^{\dagger} \Sigma h\|_2$. Then, since $h$ is i.i.d. Gaussian, this quantity concentrates to $\frac{1}{\sqrt{r}} \|(\Sigma R)^{\dagger} \Sigma h\|$, since $\|Ch\|^2 \approx \|C\|_F^2/n$ for a vector $h$ of i.i.d. $N(0,1/n)$ random variables. Finally, we can lower bound $\|(\Sigma R)^{\dagger} \Sigma\|_F^2$ by $\|(\Sigma R)^{\dagger} \Sigma RR^\top \|_F^2$ by the Pythagorean theorem, and now we have that $(\Sigma R)^{\dagger} \Sigma R$ is the identity, and so this expression is just equal to the rank of $\Sigma R$, which we prove is at least $d$. Noting that $x^* = 0$ for our instance, putting these bounds together gives
$\|x'-x^*\| \geq \sqrt{d/r}$. The last ingredient is a way to ensure that the rank of $S$ is at least $d$. Here we choose another distribution on inputs $A$ and $b$ for which it is trivial to show the rank of $S$ is at least $d$ with large probability. We require $S$ be good on the mixture. Since $S$ is fixed and good on the mixture, it is good for both distributions individually, which implies we can assume $S$ has rank $d$ in our analysis of the first distribution above. 

\subsection{Notation}
For a positive integer, let $[n] = \{1,2,\dotsc,n\}$. For a vector $x\in \mathbb{R}^n$, define $\| x\|_2=(\sum_{i=1}^n x_i^2 )^{\frac{1}{2}}$ and $\| x \|_{\infty}= \max_{i\in [n]} |x_i|$. For a matrix $A\in \mathbb{R}^{m\times n}$, define $\| A\|_2 = \sup_x \|Ax\|_2 /\|x\|_2$ to be the spectral norm of $A$ and $\| A\|_F =  ( \sum_{i,j} A_{i,j}^2  )^{1/2}$ to be the Frobenius norm of $A$. We use $A^\dagger$ to denote the Moore-Penrose pseudoinverse of $m\times n$ matrix $A$, which if $A = U \Sigma V^\top$ is its SVD (where $U\in \mathbb{R}^{m\times n}$, $\Sigma\in \mathbb{R}^{n\times n}$ and $V\in \mathbb{R}^{n\times}$ for $m\geq n$), is given by $A^{\dagger} = V \Sigma^{-1} U^\top$. 

In addition to $O(\cdot)$ notation, for two functions $f,g$, we use the shorthand $f\lesssim g$ (resp. $\gtrsim$) to indicate that $f\leq C g$ (resp. $\geq$) for an absolute constant $C$. We use $f\eqsim g$ to mean $cf\leq g\leq Cf$ for constants $c,C$.

\begin{definition}[Subspace Embedding]\label{def:subspace_embedding}
A $(1\pm\epsilon)$ $\ell_2$-subspace embedding for the column space of an $n\times d$ matrix $A$ is a matrix $S$ for which for all $x\in \mathbb{R}^d$, $\| SA x\|_2^2 = (1\pm \epsilon) \| A x\|_2^2$.
\end{definition}
\begin{definition}[Approximate Matrix Product]\label{def:approximate_matrix_product}
  Let $0<\epsilon<1$ be a given approximation parameter. Given matrices $A$ and $B$, where $A$ and $B$ each have $n$ rows, the goal is to output a matrix $C$ so that $\| A^\top B - C\|_F \leq \epsilon \| A \|_F \| B \|_F$. Typically $C$
  has the form $A^\top S^\top S B$, for a random matrix $S$ with a small
  number of rows. In particular, this guarantee holds for the subsampled randomized Hadamard transform $S$ with $O(\epsilon^{-2})$ rows \cite{DMMS11}. 
\end{definition}

\begin{figure}[!t]
  \centering
    \includegraphics[width=0.6\textwidth]{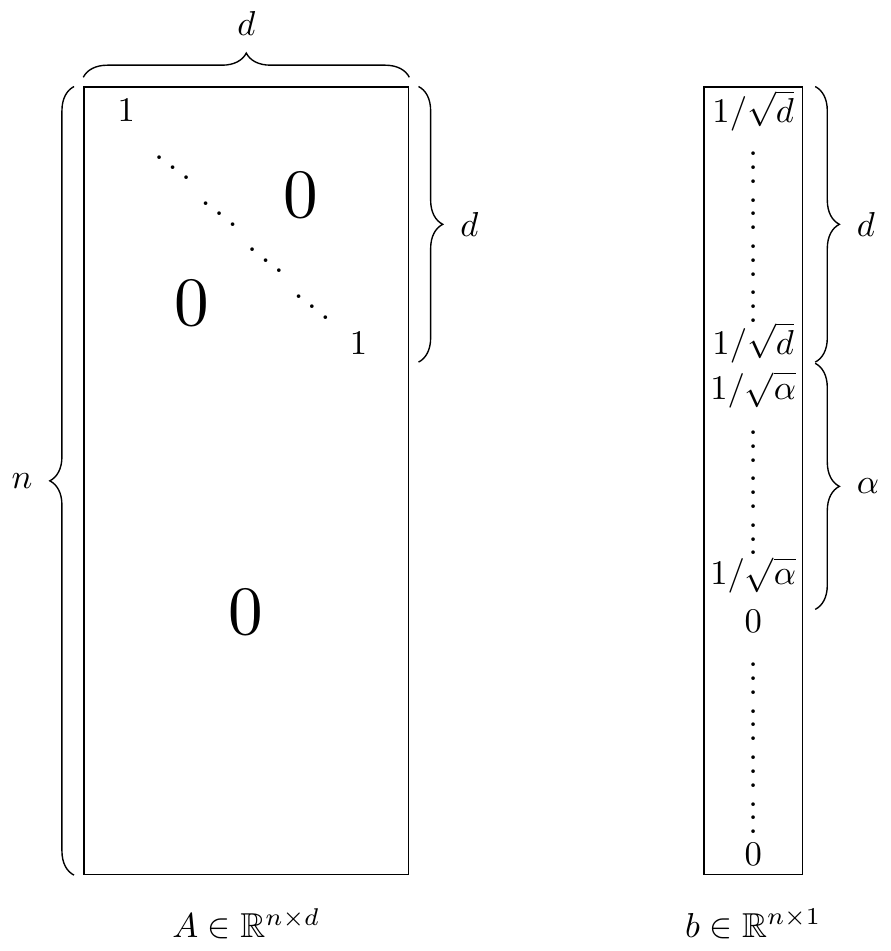}
    \caption{Our construction of $A$ and $b$ for the proof that Count-Sketch does not obey the $\ell_{\infty}$ guarantee. $\alpha < d$.}
      \label{fig:A_b_count_sketch_not_linf}
\end{figure}

\section{Warmup: Gaussians OSEs}\label{sec:gaussians}

We first show that if $S$ is a Gaussian random matrix, then it
satisfies the generalization guarantee.  This follows from the
rotational invariance of the Gaussian distribution.

\begin{theorem}\label{thm:gaussians}
  Suppose $A \in \R^{n \times d}$ has full column rank.  If the
  entries of $S \in \mathbb{R}^{m \times n}$ are i.i.d. $N(0,1/m)$, $m
  = O(d/\eps^2)$, then for any vectors $a, b$ and $x^* = A^\dagger b$, we
  have, with probability $1-1/\poly(d)$,
  \[
  |a^\top (SA)^\dagger Sb - a^\top x^*| \lesssim \frac{\eps \sqrt{\log d}}{\sqrt{d}} \norm{a}_2 \norm{b - Ax^*}_2 \norm{A^\dagger }_2.
  \]
\end{theorem}
Because $SA$ has full column rank with probability $1$, $(SA)^\dagger SA =
I$.  Therefore
\begin{align*}
|a^\top (SA)^\dagger Sb - a^\top x^*| &= |a^\top (SA)^\dagger S (b - Ax^*)| = |a^\top (SA)^\dagger S (b - AA^\dagger b)|.
\end{align*}
Thus it suffices to only consider vectors $b$ where $A^\dagger b = 0$, or
equivalently $U^\top b = 0$.  In such cases, $SU$ will be independent of
$Sb$, which will give the result.  The proof is in Appendix~\ref{app:gaussian}.

\section{SRHT Matrices}
We first provide the definition of the subsampled randomized Hadamard transform(SRHT):
let $S=\frac{1}{\sqrt{rn}} P H_n D$. Here, $D$ is an $n\times n$ diagonal matrix
with i.i.d. diagonal entries $D_{i,i}$, for which $D_{i,i}$ in uniform on
$\{-1,+1\}$. The matrix $H_n$ is 
the Hadamard matrix of size $n \times n$, and we assume $n$ is a power of $2$.
Here, $H_n=[ H_{n/2}, ~ H_{n/2}; H_{n/2}, ~ - H_{n/2} ]$ and $H_1=[1]$.
The $r\times n$ matrix $P$ samples $r$ coordinates of an $n$ dimensional vector uniformly at random.

For other subspace embeddings, we no longer have that $SU$ and $Sb$
are independent.  To analyze them, we start with a claim that allows
us to relate the inverse of a matrix to a power series.
\begin{claim}\label{claim:series}
  Let $S \in \R^{m \times n}$, $A \in \R^{n \times d}$ have SVD $A = U
  \Sigma V^\top $, and define $T \in \R^{d \times d}$ by
  \[
  T = I_d - U^\top  S^\top  S U.
  \]
  Suppose $SA$ has linearly independent columns and $\|T\|_2 \leq 1/2$. 
  Then
  \begin{align}
    (SA)^\dagger S  = V \Sigma^{-1} \left(\sum_{k = 0}^\infty T^k
      \right)U^\top  S^\top  S. \label{eq:powerseries}
  \end{align}
\end{claim}
\begin{proof}

  \begin{align*}
(SA)^\dagger S  = & ~ (A^\top S^\top SA)^{-1}A^\top S^\top S \\
 = & ~(V\Sigma U^\top S^\top SU\Sigma V^\top )^{-1}V\Sigma U^\top S^\top S \\
 = & ~ V\Sigma^{-1}(U^\top S^\top SU)^{-1} U^\top S^\top S\\
 = & ~ V\Sigma^{-1}(I_d - T)^{-1} U^\top S^\top S \\
 = & ~ V\Sigma^{-1} \left(\sum_{k=0}^\infty T^k \right) U^\top S^\top S,
  \end{align*}

where in the last equality, since $\|T\|_2 < 1$, the von Neumann
series $\sum_{k=0}^{\infty} T^k$ converges to $(I_d-T)^{-1}$. 
\end{proof}

We then bound the $k$th term of this sum:

\define{Lemma}{l:bound_T_k}{%
  Let $S \in \R^{r \times n}$ be the subsampled randomized Hadamard
  transform, and let $a$ be a unit vector. 
Then with probability $1-1/\mathrm{poly}(n)$, we have
\begin{align*}
|a^\top  S^\top  S  ( UU^\top  S^\top  S)^k b| =& O(\log^k n) \cdot  \left(  O(d(\log n)/r) + 1 \right)^\frac{k-1}{2} \cdot  (\sqrt{d} \|b\|_2 (\log n)/r + \| b\|_2 (\log^\frac{1}{2} n)/r^\frac{1}{2}) 
\end{align*}
Hence, for $r$ at least $d \log^{2k+2} n \log^2 (n/\eps) / \eps^2$, this is
at most $O(\|b\|_2 \eps / \sqrt{d})$ with probability at least
$1-1/\mathrm{poly}(n)$.
}
\state{l:bound_T_k}
We defer the proof of this lemma to the next section, and now show
how the lemma
lets us prove that SRHT matrices satisfy the generalization
bound with high probability:
\begin{theorem}\label{thm:eps}
  Suppose $A \in \R^{n \times d}$ has full column rank with $\log n =
  d^{o(1)}$.  Let $S \in \R^{m \times n}$ be a subsampled randomized
  Hadamard transform with $m = O(d^{1+\alpha}/\eps^2)$ for $\alpha = \Theta(\sqrt{\frac{\log \log n}{\log d}})$.  For any vectors $a, b$ and $x^* = A^\dagger b$, we have
  \[
  |a^\top (SA)^\dagger Sb - a^\top x^*| \lesssim \frac{\eps}{\sqrt{d}} \norm{a}_2 \norm{b - Ax^*}_2 \norm{\Sigma^{-1}}_2
  \]
  with probability $1-1/\mathrm{poly}(d)$. 
\end{theorem}

\begin{proof}
Define $\Delta = \Theta \left (\frac{1}{\sqrt{m}} \right ) (\log^c d)\norm{a}_2 \norm{b - Ax^*}_2 \norm{\Sigma^{-1}}_2 .$
For a constant $c>0$, we have that $S$ is a $(1\pm\gamma)$-subspace embedding (Definition~\ref{def:subspace_embedding}) for $\gamma =
\sqrt{\frac{d\log^c n}{m}}$ with probability $1-1/\mathrm{poly}(d)$
(see, e.g., Theorem 2.4 of \cite{woo14} and references therein), so
$\norm{SUx}_2 = (1\pm \gamma)\norm{Ux}_2$ for all $x$, which we condition
on. Hence for $T = I_d - U^\top  S^\top  S U$, we have $\norm{T}_2 \leq
(1+\gamma)^2-1 \lesssim \gamma$. In particular, $\norm{T}_2 < 1/2$ and
we can apply Claim \ref{claim:series}.

As in Section~\ref{sec:gaussians}, $SA$ has full column rank if $S$ is
a subspace embedding, so $(SA)^\dagger SA
= I$ and we may assume $x^* = 0$ without loss of generality.

By the approximate matrix product (Definition \ref{def:approximate_matrix_product}), 
we have for
some $c$ that 
\begin{align}
  |a^\top V\Sigma^{-1} U^\top  S^\top S b| \leq
  \frac{\log^c d}{\sqrt{m}} \norm{a}_2 \norm{b}_2 \norm{\Sigma^{-1}}_2 \leq \Delta\label{eq:k0}
\end{align}
with $1-1/\mathrm{poly}(d)$ probability.  Suppose this event occurs, bounding the $k=0$ term
of~\eqref{eq:powerseries}. Hence it suffices to show that the $k \geq
1$ terms of~\eqref{eq:powerseries} are bounded by $\Delta$.

By approximate matrix product (Definition~\ref{def:approximate_matrix_product}), we also have with $1-1/d^2$ probability that
\[
\norm{U^\top S^\top Sb}_F \leq \frac{\log^c d}{\sqrt{m}} \norm{U^\top }_F \norm{b}_2 \leq  \frac{\log^c d \sqrt{d}}{\sqrt{m}} \norm{b}_2.
\]
Combining with $\norm{T}_2 \lesssim \gamma$ we have for any $k$ that
\[
|a^\top V\Sigma^{-1}T^k U^\top S^\top Sb| \lesssim \gamma^k (\log^c d) \frac{\sqrt{d}}{\sqrt{m}} \norm{a}_2 \norm{\Sigma^{-1}}_2 \norm{b}_2.
\]
Since this decays exponentially in $k$ at a rate of $\gamma < 1/2$,
the sum of all terms greater than $k$ is bounded by the $k$th term.
As long as
\begin{align}
  m \gtrsim \frac{1}{\eps^2}d^{1 + \frac{1}{k}} \log^c n,\label{eq:2}
\end{align}
we have $\gamma = \sqrt{\frac{d\log^c n}{m}} < \eps d^{-1/(2k)} / \log^c n$, so that
\[
\sum_{k' \geq k} |a^\top V\Sigma^{-1}T^{k'} U^\top S^\top Sb| \lesssim \frac{\eps}{\sqrt{d}} \norm{a}_2 \norm{\Sigma^{-1}}_2 \norm{b}_2.
\]

On the other hand, by Lemma~\ref{l:bound_T_k}, increasing $m$ by a $C^k$ factor, we have for all $k$ that

\[
|a^\top  V^\top  \Sigma^{-1} U^\top  S^\top  S  ( UU^\top  S^\top  S)^k b| \lesssim \frac{1}{2^k} \frac{\eps}{\sqrt{d}}\norm{a}_2 \norm{b}_2 \norm{\Sigma^{-1}}_2
\]
with probability at least $1 - 1/\poly(d)$, as long as $m \gtrsim d
\log^{2k+2}n\log^2(d/\eps) /\eps^2$.  Since the $T^k$ term can be
expanded as a sum of $2^k$ terms of this form, we get that
\[
\sum_{k'=1}^k |a^\top V\Sigma^{-1}T^k U^\top S^\top Sb| \lesssim \frac{\eps}{\sqrt{d}}\norm{a}_2 \norm{b}_2 \norm{\Sigma^{-1}}_2
\]
with probability at least $1 - 1/\poly(d)$, as long as $m \gtrsim d
(C\log n)^{2k+2}\log^2(d/\eps) /\eps^2$ for a sufficiently large constant $C$.
Combining with~\eqref{eq:2}, the result holds as long as
\[
m \gtrsim \frac{d \log^c n}{\eps^2}\max((C\log n)^{2k+2}, d^{\frac{1}{k}})
\]
for any $k$.  Setting $k = \Theta(\sqrt{\frac{\log d}{\log \log n}})$ gives the result.
\end{proof}

{\bf{Combining Different Matrices}.} In some cases it can make sense to combine different matrices that
satisfy the generalization bound.

\define{Theorem}{thm:combine}{
  Let $A \in \R^{n \times d}$, and let $R \in \R^{m \times r}$ and
  $S \in \R^{r \times n}$ be drawn from distributions of matrices that
  are $\eps$-approximate OSEs and satisfy the generalization
  bound~\eqref{eq:ellinf}.  Then $RS$ satisfies the generalization
  bound with a constant factor loss in failure probability and
  approximation factor.
}
\state{thm:combine}
We defer the details to Appendix~\ref{sec:combining}.

\section{Proof of Lemma~\ref{l:bound_T_k}}

\begin{proof}
  Each column $S_i$ of the subsampled randomized Hadamard transform has the same distribution
  as $\sigma_i S_i$, where $\sigma_i$ is a random sign. It also has
  $\inner{S_i, S_i} = 1$ for all $i$ and $\abs{\inner{S_i, S_j}}
  \lesssim \frac{\sqrt{\log (1/\delta)}}{\sqrt{r}}$ with probability
  $1-\delta$, for any $\delta$ and $i \neq j$. 
  See, e.g., \cite{ldfu13}. 

By expanding the following product into a sum, and rearranging terms, we obtain
{
\begin{align*}
& a^\top S^\top S  ( U U^\top S^\top S)^k b \\ 
 = &\sum_{i_0,j_0, i_1, j_1, \cdots, i_k,j_k}a_{i_0} b_{j_k} \sigma_{i_0} \sigma_{i_1} \cdots \sigma_{i_k} \sigma_{j_0} \sigma_{j_1} \cdots \sigma_{j_k} \\
\cdot & \langle S_{i_0}, S_{j_0} \rangle (UU^\top )_{j_0,i_1}  \langle S_{i_1}, S_{j_1} \rangle \cdots (UU^\top )_{j_{k-1},i_k}  \langle S_{i_k}, S_{j_k} \rangle\\
= & \sum_{i_0,j_k} a_{i_0} b_{j_k} \sigma_{i_0} \sigma_{j_k} \sum_{j_0, i_1, j_1, \cdots, i_k} \sigma_{i_1} \cdots \sigma_{i_k} \sigma_{j_0} \sigma_{j_1} \cdots \sigma_{j_{k-1}} \\
& \cdot   \langle S_{i_0}, S_{j_0} \rangle (UU^\top )_{j_0,i_1}  \langle S_{i_1}, S_{j_1} \rangle \cdots (UU^\top )_{j_{k-1},i_k}  \langle S_{i_k}, S_{j_k} \rangle \\
 = & \sum_{i_0,j_k} \sigma_{i_0} \sigma_{j_k} Z_{i_0,j_k}
\end{align*}
}
where $Z_{i_0,j_k}$ is defined to be
{
\begin{eqnarray*}
Z_{i_0,j_k} &=& a_{i_0} b_{j_k} \sum_{ \substack{ i_1,\cdots i_k \\ j_0, \cdots j_{k-1} }} \prod_{c=1}^{k} \sigma_{i_c} \prod_{c=0}^{k-1} \sigma_{j_c}  \cdot \prod_{c=0}^k \langle S_{i_c}, S_{j_c}\rangle \prod_{c=1}^k (UU^\top )_{i_{c-1},j_c}
\end{eqnarray*}
}
Note that $Z_{i_0,j_k}$ is independent of $\sigma_{i_0}$ and $\sigma_{j_k}$. We observe that in the above expression if $i_0 = j_0$, $i_1 = j_1$, $\cdots$, $i_k = j_k$, 
then the sum over these indices equals $a^\top  (UU^\top )\cdots (UU^\top ) b =0$, since
$\langle S_{i_c}, S_{j_c} \rangle = 1$ in this case for all $c$. 
Moreover, the sum over all indices conditioned on $i_k = j_k$ is equal to $0$. 
Indeed, in this case, the expression can be factored into the form $\zeta \cdot U^\top  b$, 
for some random variable $\zeta$, but $U^\top  b = 0$. 

Let $W$ be a matrix with $W_{i,j} = \sigma_i \sigma_j Z_{i,j}$. We
need Khintchine's inequality:
\begin{fact}[Khintchine's Inequality] \label{fact:Khintchine}
Let $\sigma_1, \ldots, \sigma_n$ be i.i.d. sign random variables, and let $z_1, 
\ldots, z_n$ be real numbers. Then there are constants $C, C' > 0$ so that
\begin{align*}
\Pr \left[ \left|\sum_{i=1}^n z_i \sigma_i \right| \geq Ct \|z\|_2 \right] \leq e^{-C't^2}.
\end{align*}
\end{fact}
We note that Khintchine's inequality sometimes refers to bounds on the moment of
$|\sum_i z_i \sigma_i|$, though the above inequality follows readily by applying a Markov
bound to the high moments. 

We apply Fact \ref{fact:Khintchine} to each column of $W$, so that if $W_i$ is the $i$-th column, we have by a union bound that with probability $1-1/\mathrm{poly}(n)$, $\| W_i\|_2 = O( \| Z_i \|_2 \sqrt{\log n})$ simultaneously for all columns $i$. It follows that with the same probability, $\| W\|_F^2 = O(\|Z\|_F^2 \log n)$, that is, $\| W\|_F = O(\| Z\|_F \sqrt{\log n})$. We condition
on this event in the remainder. 

Thus, it remains to bound $\| Z\|_F$. By squaring $Z_{i_0, j_0}$ and
using that ${\bf E}[\sigma_i \sigma_j] = 1$ if $i = j$ and $0$
otherwise, we have, {
\begin{eqnarray}\label{eqn:Zone}
\underset{\sigma}{\bf E} [ Z_{i_0,j_k}^2 ] = a_{i_0}^2 b_{j_k}^2 \sum_{ \substack{i_1,\cdots i_k \\ j_0, \cdots j_{k-1}} } \prod_{c=0}^k \langle S_{i_c}, S_{j_c}\rangle^2 \prod_{c=1}^k (UU^\top )_{i_{c-1},j_c}^2
\end{eqnarray}
}
We defer to Appendix \ref{sec:Z}
the proof that
{
\begin{align*} \underset{S}{\bf E} [\| Z\|_F^2 ] &\leq \left( O(d(\log n)/r) + 1
  \right)^{k-1}  \cdot  (d\|b\|_2^2 (\log^2 n)/r^2 + \| b\|_2^2 (\log n)/r)
\end{align*}
}
Note that we also have the bound:
{
\begin{align*}
(O(d(\log n) /r) + 1)^{k-1} &\leq ( e^{O( d(\log n) /r)} )^{k-1} \leq e^{O(kd(\log n) /r)}  \leq O(1) 
\end{align*}
}
for any $r = \Omega(kd\log n)$.

Having computed the expectation of $\|Z\|_F^2$, we now would like to show concentration. Consider a specific 
{
\begin{eqnarray*}
Z_{i_0,j_k} = a_{i_0} b_{j_k} \sum_{i_k} \sigma_{i_k} \langle S_{i_k}, S_{j_k} \rangle \cdots \sum_{j_1} \sigma_{j_1} (UU^\top )_{j_1,i_2} \sum_{i_1} \sigma_{i_1} \langle S_{i_1}, S_{j_1}\rangle  \sum_{j_0} \sigma_{j_0} \langle S_{i_0}, S_{j_0} \rangle (UU^\top )_{j_0, i_1}.
\end{eqnarray*}
}
By Fact \ref{fact:Khintchine}, 
for each fixing of $i_1$, with probability $1-1/\mathrm{poly}(n)$, we have
{
\begin{align}\label{eqn:induct1}
&\sum_{j_0} \sigma_{j_0} \langle S_{i_0}, S_{j_0} \rangle (UU^\top )_{j_0, i_1} = O(\sqrt{\log n}) \left(\sum_{j_0}  \langle S_{i_0}, S_{j_0} \rangle^2 (UU^\top )_{j_0, i_1}^2 \right)^{\frac{1}{2}}.
\end{align}
}
Now, we can apply Khintchine's inequality for each fixing of $j_1$, and combine
this with (\ref{eqn:induct1}). 
With probability $1-1/\mathrm{poly}(n)$, again we have
{
\begin{eqnarray}
& &\sum_{i_1} \sigma_{i_1} \langle S_{i_1}, S_{j_1}\rangle \sum_{j_0} \sigma_{j_0} \langle S_{i_0}, S_{j_0} \rangle (UU^\top )_{j_0, i_1} \notag \\
& =& \sum_{i_1} \sigma_{i_1} \langle S_{i_1}, S_{j_1}\rangle O(\sqrt{\log n}) \left(\sum_{j_0}  \langle S_{i_0}, S_{j_0} \rangle^2 (UU^\top )_{j_0, i_1}^2  \right)^{\frac{1}{2}} \notag \\
& =&O(\log n )  \left( \sum_{i_1}  \langle S_{i_1}, S_{j_1}\rangle^2   \sum_{j_0}  \langle S_{i_0}, S_{j_0} \rangle^2 (UU^\top )_{j_0, i_1}^2  \right)^{\frac{1}{2}} \notag 
\end{eqnarray}
}
Thus, we can apply
Khintchine's inequality recursively over all the $2k$ indexes $j_0, i_1, j_1, \cdots ,j_{k-1}, i_k$, from which it follows that with probability $1-1/\mathrm{poly}(n)$, for each
such $i_0, j_k$, we have $Z_{i_0,j_k}^2 = O(\log^k n) \underset{S}{\bf E}[Z_{i_0, j_k}^2]$, using
(\ref{eqn:Zone}). We thus have with this probability, that 
$\|Z\|_F^2 = O(\log^k n) \underset{S}{\bf E}[\|Z\|_F^2],$
completing the proof.  
\end{proof}

\section{Lower bound for $\ell_2$ and $\ell_{\infty}$ guarantee}\label{sec:lower_bound}

We prove a lower bound for the $\ell_2$ guarantee, which immediately
implies a lower bound for the $\ell_{\infty}$ guarantee.

\begin{definition}
  Given a matrix $A\in \mathbb{R}^{n\times d}$, vector $b\in
  \mathbb{R}^{n}$ and matrix $S\in \mathbb{R}^{r\times n}$, denote
  $x^*=A^\dagger b$. We say that an algorithm ${\cal A}(A,b,S)$ that
  outputs a vector $x'=(SA)^\dagger S b$ ``succeeds'' if the following
  property holds:
$\| x' - x^* \|_2 \lesssim  \eps  \| b\|_2 \cdot \| A^\dagger \|_2 \cdot \| Ax^*-b \|_2.$
\end{definition}

\define{Theorem}{thm:l2_lower_bound}{
Suppose $\Pi$ is a distribution over $\mathbb{R}^{m\times n}$ with the property that for any $A\in \mathbb{R}^{n\times d}$ and $b\in \mathbb{R}^{n}$,
$\underset{S\sim \Pi}{ \Pr } [ {\cal A}(A,b,S) \mathrm{~succeeds~} ] \geq 19/20.$
Then $m \gtrsim \min(n,d/\eps^2)$.
}
\state{thm:l2_lower_bound}

\begin{proof}
The proof uses Yao's minimax principle. Let ${\cal D}$ be an arbitrary distribution over $\mathbb{R}^{n\times (d+1)}$, then
$
\underset{ (A,b) \sim {\cal D} }{ \mathbb{E} } ~\underset{ S \sim \Pi }{ \mathbb{E} } [ {\cal A}(A,b,S) \mathrm{~succeeds~} ] \geq 1-\delta.
$
Switching the order of probabilistic quantifiers, an averaging argument implies
the existence of a fixed matrix $S_0 \in \mathbb{R}^{m\times n}$ such that
\begin{align*}
\underset{ (A,b) \sim {\cal D} }{ \mathbb{E} }  [ {\cal A}(A,b,S_0) \mathrm{~succeeds~} ] \geq 1-\delta.
\end{align*}
Thus, we must construct a distribution ${\cal D}_{\hard}$ such that 
\begin{align*}
\underset{ (A,b) \sim {\cal D}_{\hard} }{ \mathbb{E} }  [ {\cal A}(A,b,S_0) \mathrm{~succeeds~} ] \geq 1- \delta,
\end{align*}
cannot hold for any $\Pi_0 \in \mathbb{R}^{m\times n}$ which does not satisfy $m=\Omega(d/\eps^2)$.
The proof can be split into three parts. 
 First, we prove a useful property. Second, we prove a lower bound for the case $\rank(S) \geq d$. Third, we show why $\rank(S)\geq d$ is necessary.

(\RN{1})
We show that  $[SA,Sb]$ are independent Gaussian, if both $[A,b]$ and $S$ are orthonormal matrices. We can rewrite $SA$ in the following sense,
\begin{eqnarray}\label{eq:orthonormal_S_and_A_looks_Gaussian}
\underbrace{S}_{m \times n} \cdot \underbrace{A}_{n \times d} = \underbrace{S}_{m \times n} \underbrace{ R}_{n\times n} \underbrace{R^\top}_{n\times n} \underbrace{A}_{n\times d}  =  S  \begin{bmatrix} S^\top & \overline{S}^\top \end{bmatrix} \begin{bmatrix} S \\ \overline{S} \end{bmatrix} A \notag = \begin{bmatrix} I_m & 0 \end{bmatrix}  \begin{bmatrix} S \\ \overline{S} \end{bmatrix} A  = \begin{bmatrix} I_m & 0 \end{bmatrix} \underbrace{\wt{A}}_{n \times d} =  \underbrace{\wt{A}_m}_{m \times d}   
\end{eqnarray}
where $\ov{S}$ is the complement of the orthonormal basis $S$, $I_m$ is a $m\times m$ identity matrix, and $\wt{A}_m$ is the left $m\times d$ submatrix of $\wt{A}$. Thus, using \cite{J06} as long as $m = o(\sqrt{n})$ (because of $n=\Omega(d^3)$) 
 the total variation distance between $[SA, Sb]$ and a random Gaussian matrix is small, i.e., 
\begin{equation}\label{eq:total_variation_distance}
D_{TV}( [SA,Sb], H) \leq 0.01
\end{equation}
where each entry of $H$ is i.i.d. Gaussian ${\cal N}(0,1/n)$.

(\RN{2}) 
Here we prove the theorem in the case when $S$ has rank $r\geq d$ (we will prove this is necessary in part \RN{3}. Writing $S=U \Sigma V^\top$ in its SVD, we have 
\begin{equation}\label{eq:SA_is_USigmaG}
\underbrace{S}_{m \times n}A = \underbrace{U}_{m\times r} \underbrace{\Sigma}_{r\times r} \underbrace{V^\top}_{r\times n} R R^\top A = U \Sigma G
\end{equation}
where $R=\begin{bmatrix} V & \ov{V} ~\end{bmatrix}$. By a similar argument in Equation (\ref{eq:orthonormal_S_and_A_looks_Gaussian}), as long as $r=o(\sqrt{n})$ we have that $G$ also can be approximated by a Gaussian matrix, where each entry is sampled from i.i.d. ${\cal N}(0,1/n)$. Similarly, $Sb = U \Sigma h$, where $h$ also can be approximated by a Gaussian matrix, where each entry is sampled from i.i.d. ${\cal N}(0,1/n)$.

Since $U$ has linearly independent columns, $(U\Sigma G)^\dagger U \Sigma h= (\Sigma G)^\dagger U^\top U \Sigma h = (\Sigma G)^\dagger \Sigma h$.

The $r\times d$ matrix $G$ has ${\it SVD}$ $G= \underbrace{R}_{r\times d} \underbrace{\wt{\Sigma}}_{d\times d} \underbrace{T}_{d\times d}$, and applying the pseudo-inverse property again, we have
\begin{align*}
\| (SA)^\dagger Sb\|_2 = & ~ \| (\Sigma G)^\dagger \Sigma h\|_2 = \| (\Sigma R \wt{\Sigma} T)^\dagger \Sigma h\|_2 = \| T^\dagger (\Sigma R \wt{\Sigma})^\dagger \Sigma h\|_2 = \|  (\Sigma R \wt{\Sigma} )^\dagger \Sigma h\|_2 \\
=  & ~\| \wt{\Sigma}^\dagger (\Sigma R  )^\dagger \Sigma h\|_2,
\end{align*}
where the the first equality follows by Equation (\ref{eq:SA_is_USigmaG}), the second equality follows by the {\it SVD} of $G$, the third and fifth equality follow by properties of the pseudo-inverse\footnote{\url{https://en.wikipedia.org/wiki/Moore-Penrose_pseudoinverse}} when $T$ has orthonormal rows and $\widetilde{\Sigma}$ is a diagonal matrix, and the fourth equality follows since $\| T^\dagger\|_2 = 1$ and $T$ is an orthonormal basis.

Because each entry of $G=R\wt{\Sigma}T \in \mathbb{R}^{r\times d}$ is sampled from an i.i.d. Gaussian ${\cal N}(0,1)$, using the result of \cite{V10} we can give an upper bound for the maximum singular value of $G$: $\| \wt{\Sigma} \| \lesssim \sqrt{\frac{r}{n}}$ with probability at least $.99$. Thus,
\begin{align*}
  \| \wt{\Sigma }^\dagger (\Sigma R)^\dagger \Sigma h\|_2 \geq  \sigma_{\min} (\wt{\Sigma }^\dagger) \cdot \| (\Sigma R)^\dagger \Sigma h\|_2 =  \frac{1}{\sigma_{\max} (\wt{\Sigma})} \| (\Sigma R)^\dagger \Sigma h\|_2 \gtrsim  \sqrt{n/r} \| (\Sigma R)^\dagger \Sigma h\|_2.
\end{align*}
Because $h$ is a random Gaussian vector which is independent of $(\Sigma R)^\dagger \Sigma$, by Claim \ref{cla:EAg_is_fnorm_A},
$\E_h [ \| (\Sigma R)^\dagger \Sigma h \|_2^2 ] = \frac{1}{n} \cdot \| (\Sigma R)^\dagger \Sigma \|_F^2,$
where each entry of $h$ is sampled from i.i.d. Gaussian ${\cal N}(0,1/n)$.
Then, using the Pythagorean~Theorem,
\begin{align*}
\| (\Sigma R)^\dagger \Sigma \|_F^2 = & ~ \| (\Sigma R)^\dagger \Sigma R R^\top \|_F^2 + \| (\Sigma R)^\dagger \Sigma (I-R R^\top ) \|_F^2 \\
\geq & ~ \| (\Sigma R)^\dagger \Sigma R R^\top \|_F^2 \\
=  & ~ \| (\Sigma R)^\dagger \Sigma R \|_F^2  \\
=  & ~  \rank(\Sigma R) \\
=  & ~ \rank(SA) \\
= & ~ d.
\end{align*}
Thus, $\| x' -x^*\|_2 \gtrsim \sqrt{d/r} \geq \sqrt{d/m}=\eps$.

(\RN{3}) Now we show that we can assume that $\rank(S)\geq d$.

We sample $A,b$ based on the following distribution ${\cal D}_{\hard}$: with probability $1/2$, $A,b$ are sampled from ${\cal D}_1$; with probability $1/2$, $A,b$ are sampled from ${\cal D}_2$. In distribution ${\cal D}_1$, $A$ is a random orthonormal basis and $d$ is always orthogonal to $A$. In distribution ${\cal D}_2$, $A$ is a $d\times d$ identity matrix in the top-$d$ rows and $0$s elsewhere, while $b$ is a random unit vector. 
Then, for any $(A,b)$ sampled from ${\cal D}_1$, $S$ needs to work with probability at least $9/10$. Also for any $(A,b)$ sampled from ${\cal D}_2$, $S$ needs to work with probability at least $9/10$. The latter two statements follow
since overall $S$ succeeds on ${\cal D}_{\hard}$ with probability at least $19/20$. 

Consider the case where $A, b$ are sampled from distribution ${\cal D}_2$. Then $x^*=b$ and $\OPT = 0$. Then consider $x'$ which is the optimal solution to $\min_x \| SAx - Sb\|_2^2$, so 
$x' = (SA)^\dagger Sb = (S_L)^\dagger S_L b$, 
where $S$ can be decomposed into two matrices $S_L\in \mathbb{R}^{r\times d}$ and $S_R\in \mathbb{R}^{r\times(n-d)}$, $S = \begin{bmatrix} S_L & S_R \end{bmatrix}$. Plugging $x'$ into the original regression problem,
$\| Ax' - b\|_2^2 = \| A (S_L)^\dagger S_L b - b \|_2^2$,
which is at most $(1+\eps)\OPT=0$. Thus $\rank(S_L)$ is $d$. Since $S_L$ is a submatrix of $S$, the rank of $S$ is also $d$.
%
%
\end{proof}

It remains to define several tools which are used in the main proof of the lower bound.
\begin{claim}\label{cla:EAg_is_fnorm_A}
  For any matrix $A\in \mathbb{R}^{n\times d}$, if each entry of a
  vector $g\in \mathbb{R}^d$ is chosen from an i.i.d Gaussian ${\cal
    N}(0,\sigma^2)$, then $\underset{g}{\E} [\| A g\|_2^2] = \sigma^2 \|
  A\|_F^2$ .
\end{claim}
\begin{proof}
\begin{align*}
\underset{g}{\E} [\| A g\|_2^2] = & ~ \underset{g}{\E} \left[ \sum_{i=1}^n (\sum_{j=1}^d A_{ij} g_j)^2 \right]  \\
= & ~\underset{g}{\E} \left[ \sum_{i=1}^n ( \sum_{j=1}^d A_{ij}^2 g_{j}^2 + \sum_{j\neq j'} A_{ij} A_{ij'} g_j g_{j'} ) \right] \\
= & ~\sum_{i=1}^n \sum_{j=1}^d A_{ij}^2 \sigma^2 \\
= & ~\sigma^2 \| A\|_F^2.
\end{align*}
\end{proof}

Let $g_1, g_2, \cdots, g_t$ be i.i.d. ${\cal N}(0,1)$ random variables. The random variables $\sum_{i=1}^t g_i^2$ are ${\cal X}^2$ with $t$ degree of freedom. Furthermore, the following tail bounds are known.
\begin{fact}[Lemma 1 of \cite{LM00}]\label{fac:kai_squared_distribution}
Let $g_1, g_2, \cdots, g_t$ be i.i.d. ${\cal N}(0,1)$ random variables. Then for any $x\geq 0$,
\begin{align*}
\Pr \left[ \sum_{i=1}^t g_i^2 \geq t+ 2 \sqrt{tx} + 2x \right] \leq \exp(-x),
\end{align*}
and
\begin{align*}
\Pr\left[ \sum_{i=1}^t g_i^2 \leq t- 2 \sqrt{tx}  \right] \leq \exp(-x).
\end{align*}
\end{fact}

\begin{definition}
Given a matrix $A\in \mathbb{R}^{n\times d}$, vector $b\in \mathbb{R}^{n}$ and matrix $S\in \mathbb{R}^{r\times n}$, denote $x^*=A^\dagger b$. We say that an algorithm ${\cal B}(A,b,S)$ that outputs a vector $x'=(SA)^\dagger S b$ ``succeeds'' if the following property holds:
\begin{equation*}
\| x' - x^* \|_{\infty} \lesssim  \frac{\eps}{\sqrt{d}}  \| b\|_2 \cdot \| A^\dagger \|_2 \cdot \| Ax^*-b \|_2.
\end{equation*}
\end{definition}

Applying $\| x'-x\|_{\infty} \geq \frac{1}{\sqrt{d}} \| x'-x\|_2$ to Theorem \ref{thm:l2_lower_bound} ,we obtain the $\ell_{\infty}$ lower bound as a corollary,
\begin{corollary}\label{cor:linf_lower_bound}
Suppose $\Pi$ is a distribution over $\mathbb{R}^{m\times n}$ with the property that for any $A\in \mathbb{R}^{n\times d}$ and $b\in \mathbb{R}^{n}$,
\begin{equation*}
\underset{S\sim \Pi}{ \Pr } [ {\cal B}(A,b,S) \mathrm{~succeeds~} ] \geq 9/10.
\end{equation*}
Then $m \gtrsim \min(n,d/\eps^2)$.
\end{corollary}

\newpage
\bibliographystyle{alpha}
\bibliography{ref}
\newpage
\appendix

\section*{Appendix}

\section{Proof for Gaussian case}\label{app:gaussian}
\begin{lemma}\label{lem:gaussians}
  If the entries of $S \in \mathbb{R}^{m \times n}$ are i.i.d. $N(0,1/m)$, $m = O(d/\eps^2)$, and $U^\top b = 0$, then
  \[
  |a^\top(SA)^\dagger Sb| \lesssim \frac{\eps \sqrt{\log d}}{\sqrt{d}}\norm{a}_2\norm{b}_2\norm{\Sigma^{-1}}_2
  \]
  for any vectors $a, b$ with probability $1-1/\mathrm{poly}(d)$. 
\end{lemma}
\begin{proof}
With probability $1$, the matrix $SA$ has linearly independent columns, and so $(SA)^\dagger$ is{
\begin{align*} = & ~(A^\top S^\top S A)^{-1} A^\top S^\top  \\
 = & ~( V \Sigma U^\top S^\top S U \Sigma V^\top)^{-1} V \Sigma U^\top S^\top \\
 = & ~ V \Sigma^{-1} (U^\top S^\top S U)^{-1} \Sigma^{-1} V^\top V \Sigma U^\top S^\top \\
 = & ~ V \Sigma^{-1} (U^\top S^\top S U)^{-1} U^\top S^\top. 
\end{align*}}
Hence, we would like to bound
  \[
  X = a^\top V\Sigma^{-1}(U^\top S^\top SU)^{-1} U^\top S^\top Sb.
  \]
  It is well-known (stated, for example, explicitly in 
  Theorem 2.3 of \cite{woo14}) that with probability
  $1-\exp(-d)$, the singular values of $SU$ are $(1 \pm \eps)$ for $m
  = O(d/\eps^2)$.  We condition on this event.
  It follows that {
\begin{eqnarray*}
 && \|V \Sigma^{-1} (U^\top S^\top S U)^{-1} U^\top S\|_2 \\
& = & \|\Sigma^{-1} (U^\top S^\top S U)^{-1} U^\top S\|_2\\
& \leq & \|\Sigma^{-1}\|_2 \|(U^\top S^\top S U)^{-1}\|_2 \|U^\top S\|_2\\
& \leq & \|\Sigma^{-1}\|_2 \cdot \frac{1}{1-\eps} \cdot (1+\eps)\\
& = & O(\|\Sigma^{-1}\|_2),
\end{eqnarray*}}
where the first equality uses that $V$ is a rotation, the first inequality follows by sub-multiplicativity,
and the second inequality uses that the singular values of $SU$ are in the range $[1 -\eps, 1+\eps]$. 
Hence, with probability $1-\exp(-d)$,
\begin{eqnarray}\label{eqn:oper}
\|a^\top V \Sigma^{-1} (U^\top S^\top SU)^{-1} U^\top S^\top\|_2 = O( \|\Sigma^{-1}\|_2 \|a\|_2).
\end{eqnarray}

The main observation is that since $U^\top b = 0$, $SU$ is statistically independent from $Sb$. 
Hence, $Sb$ is distributed as $N(0, \|b\|_2^2 I_m)$, conditioned on the vector $a^\top V \Sigma^{-1} (U^\top S^\top SU)^{-1} U^\top S^\top$. 
It follows that conditioned on the value of $a^\top V \Sigma^{-1} (U^\top S^\top SU)^{-1} U^\top S^\top$, $X$ is distributed as 
\begin{align*}
N(0, \|b\|_2^2 \| a^\top V \Sigma^{-1} (U^\top S^\top SU)^{-1} U^\top S^\top \|_2^2/m),
\end{align*} 
and so using (\ref{eqn:oper}) , with probability $1-1/\mathrm{poly}(d)$, 
we have $|X| = O(\eps \sqrt{\log d} \|a\|_2 \|b\|_2 \|\Sigma^{-1}\|_2/ \sqrt{d})$.
\end{proof}

\section{Combining Different Matrices}\label{sec:combining}

In some cases it can make sense to combine different matrices that
satisfy the generalization bound.

\restate{thm:combine}

\begin{proof}
  For any vectors $a, b$, and $x^* = A^\dagger b$ we want to show
  \[
  |a^\top (RSA)^\dagger RSb - a^\top x^*| \lesssim \frac{\eps}{\sqrt{d}}\norm{a}_2 \norm{b - Ax^*}_2 \norm{A^\dagger }_2 
  \]
  As before, it suffices to consider the $x^*=0$ case.  We have with
  probability $1-\delta$ that
  \[
  |a^\top (SA)^\dagger Sb| \lesssim \frac{\eps}{\sqrt{d}}\norm{a}_2 \norm{b}_2 \norm{A^\dagger }_2 ;
  \]
  suppose this happens.  We also have by the properties of $R$,
  applied to $SA$ and $Sb$, that
  \[
  |a^\top (RSA)^\dagger RSb - a^\top (SA)^\dagger Sb| \lesssim \frac{\eps}{\sqrt{d}}\norm{a}_2 \norm{Sb}_2 \norm{(SA)^\dagger}_2.
  \]
  Because $S$ is an OSE, we have $\norm{Sb}_2 \leq (1+\eps)$ and
  $\norm{(SA)^\dagger }_2 \gtrsim (1-\eps) \norm{A^\dagger}_2$.  Therefore
  \[
  |a^\top (RSA)^\dagger RSb| \lesssim \frac{\eps}{\sqrt{d}}\norm{a}_2 \norm{b}_2 \norm{A^\dagger}_2
  \]
\end{proof}

We describe a few of the applications of combining sketches.

\subsection{Removing dependence on $n$ via Count-Sketch}

One of the limitations of the previous section is that the choice of
$k$ depends on $n$.  To prove that theorem, we have to assume that
$\log d>\log \log n$. Here, we show an approach to remove that
assumption.

The main idea is instead of applying matrix $S\in \mathbb{R}^{m\times
  n}$ to matrix $A\in \mathbb{R}^{n\times d}$ directly, we pick two
matrices $ S \in \mathbb{R}^{m\times \mathrm{poly}(d) }$ and $C\in
\mathbb{R}^{\mathrm{poly}(d) \times n}$, e.g. $S$ is FastJL matrix and
$C$ is Count-Sketch matrix with $s=1$. We first compute $C\cdot A$,
then compute $S\cdot (CA)$. The benefit of these operations is $S$
only needs to multiply with a matrix $(CA)$ that has
$\mathrm{poly}(d)$ rows, thus the assumption we need is $\log d > \log
\log (\mathrm{poly} (d))$ which is always true. The reason for
choosing $C$ as a Count-Sketch matrix with $s=1$ is: (1)
$\mathrm{nnz}(CA) \leq \mathrm{nnz}(A)$ (2) The running time is
$O(\mathrm{poly}(d)\cdot d + \mathrm{nnz}(A))$.

\subsection{Combining Gaussians and SRHT} By combining Gaussians
with SRHT matrices, we can embed into the optimal dimension $O(d/\eps^2)$
with fast $\Ot(nd\log n + d^\omega/\eps^4)$ embedding time.

\subsection{Combining all three} By taking Gaussians times SRHT times
Count-Sketch, we can embed into the optimal dimension $O(d/\eps^2)$
with fast $O(\text{nnz}(A) + d^4 \text{poly}(\frac{1}{\eps}, \log d))$
embedding time.

\section{Count-Sketch does not obey the $\ell_{\infty}$ guarantee}\label{sec:cs}

Here we demonstrate an $A$ and a $b$ such that Count-Sketch will not
satisfy the $\ell_\infty$ guarantee with constant probability, so such
matrices cannot satisfy the generalization guarantee~\eqref{eq:ellinf}
with high probability.

\define{Theorem}{thm:count_sketch_not_infty}{
  Let $S \in \R^{m \times n}$ be drawn as a Count-Sketch matrix with
  $s$ nonzeros per column.  There exists a matrix $A \in \R^{n \times
    d}$ and $b \in \R^n$ such that, 
  if $s^2 d \lesssim m \lesssim \sqrt{d^3 s}$,
   then the ``true''
  solution $x^* = A^\dagger b$ and the approximation $x' = (SA)^\dagger Sb$ have
  large $\ell_\infty$ distance with constant probability:
  \[
  \norm{x' - x^*}_\infty \gtrsim \sqrt{\frac{d}{ms}} \norm{b}_2.
  \]
  Plugging in $m = d^{1.5}$ and $s = d^{0.25}$ we find that
  \[
  \norm{x' - x^*}_\infty \gtrsim 1/d^{3/8}\norm{b}_2 \gg 1/\sqrt{d} \norm{b}_2,
  \]
  even though such a matrix is an OSE with probability exponential in
  $s$.  Therefore there exists a constant $c$ for which this matrix
  does not satisfy the generalization guarantee~\eqref{eq:ellinf} with
  $1 - \frac{c}{d}$ probability.
}
\state{thm:count_sketch_not_infty}

\begin{proof}
  We choose the matrix $A$ to be the identity on its top $d$ rows:
  $A= \begin{bmatrix}I_d \\ 0 \end{bmatrix}$. Choose some $\alpha \geq
  1$, set the value of the first $d$ coordinates of vector $b$ to be
  $\frac{1}{\sqrt{d}}$ and set the value to be $1 / \sqrt{\alpha}$ for
  the next $\alpha$ coordinates, with the remaining entries all
  zero. Note that $\norm{b}_2 = \sqrt{2}$, $x^* = (1/\sqrt{d}, \dotsc,
  1/\sqrt{d})$, and
$
\norm{Ax^* - b}_2 = 1.
$

Let $S_k$ denote the $k$th column vector of matrix $S \in \mathbb{R}^{m\times n}$. 
We define two events,
Event I, $\forall k'\in [d]$ and $k'\neq k $, we have $\supp(S_{k'}) \cap \supp(S_k) = \emptyset$; Event II, $\exists$ a unique $k'\in \{ d+1, d+2, \cdots, d+\alpha \}$ such that $| \supp(S_{k'}) \cap \supp (S_k) |=1$, and all other $k'$ have $\supp(S_{k'}) \cap \supp(S_k) = \emptyset$. Using Claim \ref{cla:two_events_for_S_k}, with probability at least $.99$ there exists a $k$ for which both events hold.

Given the constructions of $A$ and $b$ described early, it is obvious that {
\[Ax-b = \begin{bmatrix} x_1-\frac{1}{\sqrt{d}}, \cdots, x_d -\frac{1}{\sqrt{d}}, -\frac{1}{\sqrt{\alpha}}, \cdots, -\frac{1}{\sqrt{\alpha}}, 0,\cdots, 0 \end{bmatrix}^\top.\]}

Conditioned on event I and II are holding, then denote $\supp(S_j)= \{i_1, i_2, \cdots, i_s \}$. Consider the terms involving $x_j$ in the quadratic form
\[
\min_x \norm{SAx - Sb}_2^2.
\]
it can be written as $(s-1)(x_j-1/\sqrt{d})^2 + (x_j - 1/\sqrt{d} \pm 1/\sqrt{\alpha})^2$.
  Hence the optimal
$x'$ will have $x'_j = \frac{1}{\sqrt{d}} \pm \frac{1}{s\sqrt{\alpha}}$, which is different
from the desired $1/\sqrt{d}$ by $\frac{1}{s \sqrt{\alpha}} $.  Plugging in our requirement of
 $\alpha \eqsim m^2/(s^3d^2)$, we have  
\[ 
\norm{x' - x^*}_\infty \geq  \frac{1}{s\sqrt{\alpha}}   \gtrsim c\sqrt{\frac{sd^2}{m^2}}  \gtrsim \frac{1}{\sqrt{d}}
\]
where the last inequality follows by $m \lesssim \sqrt{sd^3}$. Thus, we get the result. 
\end{proof}

\begin{figure}[!t]
  \centering
    \includegraphics[width=0.9\textwidth]{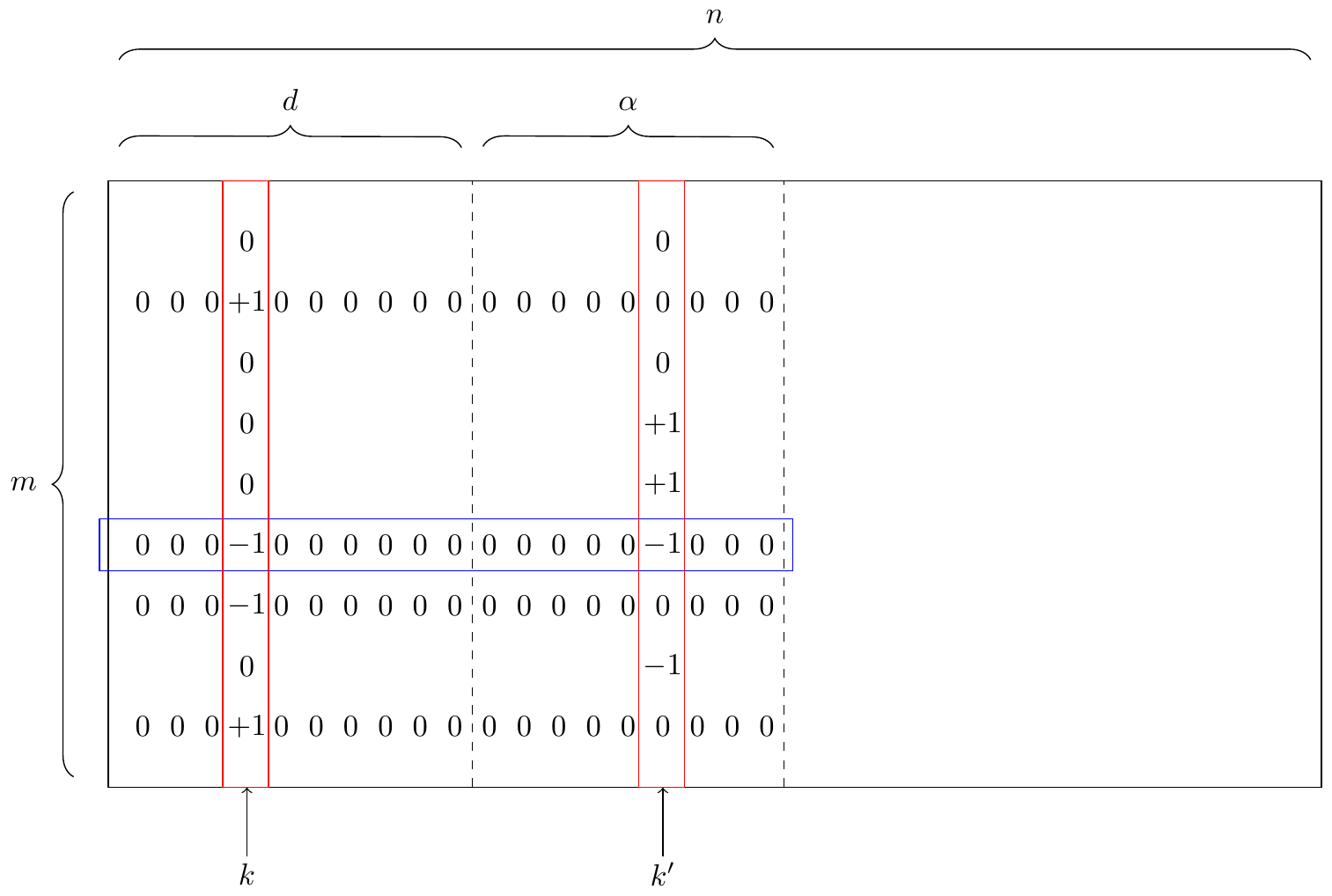}
    \caption{Count Sketch matrix $S\in\mathbb{R}^{m\times n}$. Event I, for any $k'\in [d]$ and $k'\neq k$, $\supp(S_k) \cap \supp(S_{k'}) = \emptyset$. Event II, there exists a unique $k' \in \{d+1, d+2, \cdots, d+\alpha\}$ such that $S_k$ and $S_{k'}$ intersect at exactly one location(row index). }\label{fig:S_count_sketch_not_linf}
\end{figure}

\begin{claim}\label{cla:two_events_for_S_k}
  If $ m = \Omega(s^2 d)$, $m = o(d^2)$, $\alpha <d$, and $\alpha =
  O(\frac{m^2}{s^3 d^2})$, with probability at least $.99$
  there exists a $k \in [d]$ for which both event I and II hold.
\end{claim}

\begin{proof}
If $m = \Omega(s^2 d)$, then for any $i$ in $\{1, 2, ..., d\}$, let $X_i$ be an indicator that the entries of column i are disjoint from all $i'$ in $[d] \backslash \{i\}$. Then $\E[X_i] \geq .9999$, so by Markov's inequality, with probability .99, we have $.99d$ columns having this property (indeed, the expected value of $d-X$ is at most $.0001d$, so $\Pr[ d-X \geq .01 d ] \leq \frac{\E[d-X]} {.01d} \leq \frac{.0001d}{.01d} = .01 $). 
Define Event $E$ to be that $.99d$ columns of first $d$ columns have the property that the entries of that column are disjoint from all the other $d-1$ columns. Let $S$ be the set of these $.99d$ columns. Let $N$ be the union of supports of columns in $S$.


Each column $i$ in $\{d+1, ..., d+\alpha\}$ chooses $s$ non-zero entries. Define event $F$( which is similar as event $E$) to be that $.99\alpha$ columns of the next $\alpha$ columns have the property that the entries of that column are disjoint from all the other $\alpha-1$ columns. By the same argument, since $\alpha < d$, with probability .99, we have $.99\alpha$ columns in $\{d+1, ..., d+\alpha\}$ being disjoint from other columns in $\{d+1, ..., d+\alpha\}$. Condition on event $F$ holding. Let $L$ be the multiset union of supports of all columns in $\{d+1, ..., d+\alpha\}$. Then $L$ has size $\alpha \cdot s$. Let $M$ be the union of supports of all columns in $\{d+1, ..., d+\alpha\}$, that is, the set union rather than the multiset union. Note that $|M| \geq .99\alpha \cdot s$ because of $.99 \alpha$ columns are disjoint from each other.

The intersection size $x$ of $N$ and $M$ is hyper-geometrically distributed with expectation
\begin{equation*}
\E[x] = \frac{s|S| \cdot |M|}{ m}. 
\end{equation*} 
By a lower tail bound for the hypergeometric distribution \footnote{\url{https://en.wikipedia.org/wiki/Hypergeometric_distribution}} , 
\begin{equation*}
\Pr[x \leq (p-t)n] \leq \exp(-2t^2n),
\end{equation*} 
where $p = s\cdot |S|/m$ and $n = |M|$, so 
\begin{equation*}
\Pr[x \leq \E[x] - t \cdot |M|] \leq \exp(-2t^2\cdot |M|) \leq 0.01,
\end{equation*}
where the last inequality follows by setting  $t = \Theta(1/\sqrt{|M|})$. Thus, we get with probability $.99$, the intersection size is at least $\frac{s|S| \cdot |M|}{ m} - \Theta(\sqrt{|M|})$ .

Now let $W$ be the distinct elements in $L\backslash M$, so necessarily $|W| \leq .01 \alpha \cdot s$.  By an upper tail bound for the hypergeometric distribution, the intersection size $y$ of $W$ and $N$ satisfies
\begin{equation*}
\Pr[y \geq (p+t)n] \leq \exp(-2t^2n),
\end{equation*}
where $p = s\cdot |S|/m$ and $n = |W|$, we again get
\begin{equation*}
 \Pr[y \geq \E[y] + t\cdot |W|] \leq \exp(-2t^2\cdot|W|).
\end{equation*} 
If $|W| = 0$, then $y = 0$. Otherwise, we can set $t = \Theta(1/\sqrt{|W|})$ so that this probability is less than $.01$, and we get with probability $.99$, the intersection size  $y$ is at most $s\cdot |S|\cdot |W|/m + \Theta(\sqrt{|W|})$. Note that we have that $\Theta(\sqrt{|M|})$ and $\Theta(\sqrt{|W|})$ are bounded by $\Theta(\sqrt{s \cdot \alpha})$. Setting $\alpha = O( \frac{m^2}{ s^3 d^2} )$ suffices to ensure $y$ is at most $(1.01)s\cdot |S|\cdot |W|/m$, and earlier that $x$ is at least $.99\cdot s\cdot |S|\cdot |M|/m$.

The probability one of the $|S|$ blocks in $N$ has two or more intersections with $M$ is less than ${x \choose 2}$ times the probability two random distinct items in the intersection land in the block. This probability is 
\begin{equation*}
 \frac{ {x \choose 2}\cdot {s \choose 2}} { {s\cdot |S| \choose 2}} = \Theta(x^2/|S|^2) = \Theta(x^2/d^2) = \Theta(m^2/(d^4 s^2)).
\end{equation*} 
So the expected number of such blocks is $\Theta(m^2 sd/(d^4 s^2)) = \Theta(m^2 /(d^3 s))$ which is less than $(.99\cdot s\cdot |S|\cdot |M|)/(2m) \leq X/2$ if $m = o(d^2)$, which we have. So, there are at least $x/2$ blocks which have intersection size exactly 1 with $N$. Note that the number of intersections of the $|S|$ blocks with $W$ is at most $y$, which is at most $(1.01)s\cdot |S|\cdot |W|/m \leq (1.01)s\cdot |S|\cdot \frac{1}{99}\cdot |M|/m < x/2$, and therefore there exists a block, that is, a column among the first $d$ columns, which intersects $M$ in exactly one position and does not intersect $W$. This is our desired column. Thus, we complete the proof.
\end{proof}

\section{Leverage score sampling does not obey the $\ell_\infty$ guarantee}

Not only does Count-Sketch fail, but so does leverage score sampling,
which is a technique that takes a subsample of rows of $A$ with
rescaling.  In this section we show an $A$ and a $b$ such that
leverage score sampling will not satisfy the $\ell_{\infty}$ guarantee.  We
start with a formal definition of leverage scores. 

\begin{definition}[Leverage Scores]
Given an arbitrary $n\times d$ matrix $A$, with $n>d$, let $U$ denote the $n\times d$ matrix
consisting of the $d$ left singular vectors of $A$, let $U_{(i)}$ denote the $i-$th row of the 
matrix $U$, so $U_{(i)}$ is a row vector. Then the leverage scores of the rows of $A$ are given by $l_i = \| U_{(i)}\|_2^2$, 
for $i\in [n]$.
\end{definition}
The leverage score sampling matrix can be thought of as a square diagonal matrix $D\in \mathbb{R}^{n\times n}$ with diagonal entries chosen from some distribution. If $D_{ii}=0$, it means we do not choose the $i$-th row of matrix $A$., If $D_{ii}>0$, it means we choose that row of the matrix $A$ and also rescale that row. We show that the leverage score sampling matrix cannot achieve $\ell_{\infty}$ 
guarantee, nor can it achieve our notion of generalization error.

\define{Theorem}{thm:leverage_score_not_infty}{
Let $D \in \R^{n \times n}$ be a leverage score sampling matrix with
  $m$ nonzeros on the diagonal. There exists a matrix $A \in \R^{n \times
    d}$ and a vector $b \in \R^n$ such that, if $m \lesssim d\sqrt{d}$, then the ``true''
  solution $x^* = A^\dagger b$ and the approximation $x' = (DA)^\dagger Db$ have
  large $\ell_\infty$ distance with constant probability:
  \[
  \norm{x' - x^*}_\infty \gtrsim \frac{1}{\sqrt{d}} \norm{b}_2.
  \]
    Therefore there exists a constant $c$ for which this matrix
  does not satisfy the generalization guarantee~\eqref{eq:ellinf} with
  $1 - \frac{c}{d}$ probability.
}
\state{thm:leverage_score_not_infty}

\begin{proof}
We choose the matrix $A$ to be the identity on its top $d$ rows, and $L$ scaled identity matrices $\frac{1}{ \sqrt{\alpha d}} I_d$ for the next $dL$ rows, where $L$ satisfies $\frac{1}{d} + \frac{1}{\alpha d} L =1$ (to normalize each column of $A$), which implies $L= \alpha(d-1)$.
Choose some $\beta \in [1,d)$. Set the value of the first $d$ coordinates of vector $b$ to be $\frac{1}{\sqrt{d}}$ and set the value to be $\frac{1}{\sqrt{\beta} }$ for the next $\beta$ coordinates, with the remaining entries all zero. Note that $\| b \|_2 = \sqrt{2}$.

First, we compute $\|Ax-b\|_2^2$. Because $\beta$ is less than $d$, there are two kinds of $x_j$: one involves the following term,
\begin{equation}\label{eq:nonspecial_j_without_leverage}
(\frac{1}{\sqrt{d}} x_j - \frac{1}{\sqrt{d}} )^2  + (L-1)  ( \frac{1}{\sqrt{\alpha d} }  x_j )^2,
\end{equation}
where the optimal $x_j$ should be set to $1/d$. The other involves the term:
\begin{equation}\label{eq:special_j_without_leverage}
(\frac{1}{\sqrt{d}} x_j - \frac{1}{\sqrt{d} } )^2 + (\frac{1}{\sqrt{\alpha d}} x_j - \frac{1}{ \sqrt{\beta}})^2 + (L-1) ( \frac{1}{\sqrt{\alpha d} }  x_j )^2,
\end{equation}
where the optimal $x_j$ should be set to $1/d +  1/\sqrt{\alpha \beta d} $. Because we are able to choose $\alpha, \beta$ such that $\alpha \beta \gtrsim d$, then
\begin{equation*}
x_j = 1/d +  1/\sqrt{\alpha \beta d} \lesssim 1/d.
\end{equation*} 


Second, we compute $\| DAx - Db \|_2^2$. With high probability, there exists a $j$ satisfying Equation (\ref{eq:special_j_without_leverage}), but after applying leverage score sampling, the middle term of Equation (\ref{eq:special_j_without_leverage}) is removed. Let $p_1 = \frac{1}{d}$ denote the leverage score of each of the top $d$ rows of $A$, and let $p_2=\frac{1}{\alpha d}$ denote the leverage score of each of the next $Ld$ rows of $A$. We need to discuss the cases $m>d$ and $m\leq d$ separately.

If $m > d$, then the following term involves $x_j$, {
\begin{eqnarray*}
&& ( \frac{1}{ \sqrt{p_1}} \frac{1}{\sqrt{d}} x_j - \frac{1}{ \sqrt{p_1}} \frac{1}{\sqrt{d}} )^2 + \frac{m-d}{d} \cdot  ( \frac{1}{\sqrt{p_2}} \frac{1}{\sqrt{\alpha d} }x_j )^2 \\ 
& = &\frac{1}{p_1} ( \frac{1}{\sqrt{d}} x_j - \frac{1}{\sqrt{d}} )^2 + \frac{m-d}{d} \cdot \frac{1}{p_2} (\frac{1}{\sqrt{\alpha d} }x_j )^2 \\ 
& = & d \left( ( \frac{1}{\sqrt{d}} x_j - \frac{1}{\sqrt{d}} )^2 + \frac{m-d}{d} \alpha (\frac{1}{\sqrt{\alpha d} }x_j )^2  \right).
\end{eqnarray*}
}
where the optimal $x_j$ should be set to{ 
\begin{align*}
x_j = & ~ \frac{ 1/d }{ 1/d + (m-d)\alpha/ (\alpha d^2) }  \\
=  & ~ \frac{ 1 }{ 1+ (m-d)/d } \\
\gtrsim & ~ \frac{1}{ (m-d) /d } \\
 \gg & ~  \frac{1}{\sqrt{d}}. & \text{~by~ } m  \ll d \sqrt{d}
\end{align*}}
If $m \leq d$, then the term involving $x_j$ is
$
( \frac{1}{ \sqrt{p_1}} \frac{1}{\sqrt{d}} x_j - \frac{1}{ \sqrt{p_1}} \frac{1}{\sqrt{d}} )^2   
$
where the optimal $x_j$ should be set to be $1 \gg 1/\sqrt{d}$.

Third, we need to compute $\| Ax^* -b \|_2^2$ and $\sigma_{\min}(A)$. It is easy to see that $\sigma_{\min}(A)$ because $A$ is an orthonormal matrix. The upper bound for $\| Ax^* -b \|_2^2 =2$, and the lower bound is also a constant, which can be proved in the following way: 
\begin{align*}
\| A x^* - b\|_2^2 & =  \sum_{j=1}^\beta (\ref{eq:nonspecial_j_without_leverage}) +\sum_{j=\beta+1}^d (\ref{eq:special_j_without_leverage}) \geq d (\frac{1}{\sqrt{d}} \frac{1}{d} - \frac{1}{\sqrt{d}})^2 \gtrsim  d \cdot \frac{1}{d} = 1.
\end{align*}
\end{proof}

\section{Bounding $\E[ \|Z\|_F^2 ]$}\label{sec:Z}

Before getting into the proof details, we define the key property of $S$ being used in the rest of the proofs.
\begin{definition}[All Inner Product Small(AIPS) Property]
For any matrix $S\in \mathbb{R}^{r\times n}$, if for all $i,j \in [n]$ with $i\neq j$ we have
\begin{equation*}
| \langle S_i, S_j \rangle | = O(\sqrt{\log n} /\sqrt{r}),
\end{equation*}
we say that $S$ satisfies the ``$\mathrm{AIPS}$'' property.
\end{definition}
\begin{claim}
  If $S\in \mathbb{R}^{r\times n}$ is a subsampled Hadamard transform
  matrix, then the $\mathrm{AIPS}$ property holds with probability at
  least $1-1/\poly(n)$.
\end{claim}
\begin{proof}
From the structure of $S$, for any $i\neq j$, we have with probability $1-1/\poly(n)$ such that $ |\langle S_i, S_j \rangle| = O(\sqrt{\log n} /\sqrt{r})$.  Applying a union bound over $O(n^2)$ pairs, we obtain that 
\begin{align*}
\Pr[ \text{~AIPS~holds~}] \geq 1-1/\poly(n).
\end{align*}
\end{proof}

The main idea for bounding $\E [\| Z\|_F^2]$ is to rewrite it as 
$
 \E[ \| Z \|_F^2] = \E [~ \|Z \|_F^2 ~|~ \text{AIPS~holds}~] + \E [~ \|Z \|_F^2 ~|~ \text{AIPS~does~not~hold}~].
$
Because $ \Pr[ \text{~AIPS~does~not~hold}]$ is at most $1/\poly(n)$,
the first term dominates the second term, which means we
only need to pay attention to the first term. We repeatedly apply this
idea until all the $S$ are removed.

We start by boundinbg $\E[\| Z\|_F^2]$ by squaring $Z_{i_0, j_0}$ and using that $\E[\sigma_i \sigma_j] = 1$ if $i = j$ and $0$ otherwise. Then, we obtain,  
{%
\begin{eqnarray}\label{eqn:Zone}
\underset{\sigma}{\E} [ Z_{i_0,j_k}^2 ] =  a_{i_0}^2 b_{j_k}^2 \sum_{i_1,\cdots i_k, j_0, \cdots j_{k-1}} \prod_{c=0}^k \langle S_{i_c}, S_{j_c}\rangle^2 \prod_{c=1}^k (UU^\top)_{i_{c-1},j_c}^2.
\end{eqnarray}
}
We thus have,  {%
\begin{align*}
  \sum_{i_0,j_k, j_k\neq i_k} a_{i_0}^2 \langle S_{i_0}, S_{j_0}\rangle^2 b_{j_k}^2 \langle S_{i_k}, S_{j_k} \rangle^2  =  a_{j_0}^2 \|b\|_2^2 O( (\log n)/r )  + \| a\|_2^2 \|b\|_2^2 O( (\log^2 n)/ r^2 ) \eqdef  C_{j_0},
\end{align*}}
where the first equality is from our conditioning, and the second equality is the definition of $C_{j_0}$. 
Hence, $\underset{S}{\E}  [\| Z\|_F^2 ]$ is
{%
\begin{eqnarray*}
& = & \sum_{i_1,\cdots i_k, j_0, \cdots j_{k-1}} \prod_{c=1}^{k-1} \langle S_{i_c}, S_{j_c}\rangle^2 \prod_{c=1}^k (UU^\top)_{j_{c-1},i_c}^2 \cdot \sum_{i_0,j_k, j_k\neq i_k} a_{i_0}^2 \langle S_{i_0}, S_{j_0}\rangle^2 b_{j_k}^2 \langle S_{i_k}, S_{j_k} \rangle^2 \\
& = & \sum_{i_1,\cdots i_k, j_0, \cdots j_{k-1}} \prod_{c=1}^{k-1} \langle S_{i_c}, S_{j_c}\rangle^2 \prod_{c=1}^k (UU^\top)_{j_{c-1},i_c}^2 C_{j_0} \\
& = & \sum_{i_1,\cdots i_k, j_0, \cdots j_{k-1}} \langle S_{i_{k-1}}, S_{j_{k-1}} \rangle^2 (UU^\top)_{j_{k-1},i_k}^2 \cdot \prod_{c=1}^{k-2} \langle S_{i_c}, S_{j_c}\rangle^2 \prod_{c=1}^{k-1} (UU^\top)_{j_{c-1},i_c}^2 C_{j_0}, \\
\end{eqnarray*}
}
where the first equality follows from (\ref{eqn:Zone}), the second equality by definition
of $C_{j_0}$, and the final equality by factoring out $c = k-1$ from one product and
$c = k-2$ from the other product. 

The way to bound the term $\langle S_{i_{k-1}}, S_{j_{k-1}} \rangle$ is by separating the diagonal term where $i_{k-1} = j_{k-1}$ and the non-diagonal term where $i_{k-1} \neq j_{k-1}$. We now use the aforementioned property of $S$, namely, that $\langle S_{i_{k-1}}, S_{j_{k-1}} \rangle = 1$, if $i_{k-1} = j_{k-1}$, while for $i_{k-1}\neq j_{k-1}$, we have with probability $1 - 1/\mathrm{poly}(n)$ that $|\langle S_{i_{k-1}}, S_{j_{k-1}} \rangle| = O(\sqrt{\log n}/\sqrt{r})$ conditioned on $\mathrm{AIPS}$ holding.  

Conditioned on $\mathrm{AIPS}$ holding, we can recursively reduce the number of terms in the product:
{%
\begin{eqnarray*}
&& \| Z\|_F^2 \\
& = & \sum_{ \substack{i_1,\cdots i_k, j_0, \cdots j_{k-1}, i_{k-1} \neq j_{k-1} }} O((\log n)/r) \cdot (UU^\top)_{j_{k-1},i_k}^2   \cdot \prod_{c=1}^{k-2} \langle S_{i_c}, S_{j_c}\rangle^2 \prod_{c=1}^{k-1} (UU^\top)_{j_{c-1},i_c}^2 C_{j_0} \\
& + & \sum_{i_1,\cdots i_k, j_0, \cdots j_{k-1}, i_{k-1} = j_{k-1}} 1  \cdot (UU^\top)_{j_{k-1},i_k}^2  \cdot  \prod_{c=1}^{k-2} \langle S_{i_c}, S_{j_c}\rangle^2 \prod_{c=1}^{k-1} (UU^\top)_{j_{c-1},i_c}^2 C_{j_0} \\
& \leq & \sum_{i_1,\cdots i_k, j_0, \cdots j_{k-1} } O( (\log n) / r)  \cdot (UU^\top)_{j_{k-1},i_k}^2  \cdot  \prod_{c=1}^{k-2} \langle S_{i_c}, S_{j_c}\rangle^2 \prod_{c=1}^{k-1} (UU^\top)_{j_{c-1},i_c}^2 C_{j_0} \\
& + & \sum_{i_1,\cdots i_k, j_0, \cdots j_{k-1}, i_{k-1} = j_{k-1}} 1 \cdot (UU^\top)_{j_{k-1},i_k}^2   \cdot \prod_{c=1}^{k-2} \langle S_{i_c}, S_{j_c}\rangle^2 \prod_{c=1}^{k-1} (UU^\top)_{j_{c-1},i_c}^2 C_{j_0},
\end{eqnarray*}
}
where the first equality follows from the property just mentioned, and the inequality
follows by including back the tuples of indices for which $i_{k-1} = j_{k-1}$, using
that each summand is non-negative. 

Our next step will be to bound the term $(UU^\top)^2_{j_{k-1}, i_k}$. We have, $\| Z\|_F^2$ is
{
\begin{align*}
 & \leq \sum_{i_k,j_{k-1}} (UU^\top)_{i_k,j_{k-1}}^2  \sum_{\substack{i_1,\cdots, i_{k-1} \\ j_0, \cdots, j_{k-2} } } O( (\log n) / r)  \\
&\cdot \prod_{c=1}^{k-2} \langle S_{i_c}, S_{j_c}\rangle^2 \prod_{c=1}^{k-1} (UU^\top)_{j_{c-1},i_c}^2 C_{j_0} \\
& + \sum_{\substack{ i_1,\cdots, i_k \\ j_0, \cdots, j_{k-1}\\ i_{k-1} = j_{k-1} } } 1  \cdot (UU^\top)_{j_{k-1},i_k}^2  \prod_{c=1}^{k-2} \langle S_{i_c}, S_{j_c}\rangle^2 \prod_{c=1}^{k-1} (UU^\top)_{j_{c-1},i_c}^2 C_{j_0} \\
& = \underbrace{ O(d(\log n)/r) \sum_{ \substack{ i_1,\cdots, i_{k-1}\\ j_0, \cdots, j_{k-2} } } \prod_{c=1}^{k-2} \langle S_{i_c}, S_{j_c}\rangle^2 \prod_{c=1}^{k-1} (UU^\top)_{j_{c-1},i_c}^2 C_{j_0} }_{A}  \\
& + \underbrace{ \sum_{ \substack{ i_1,\cdots, i_k \\j_0, \cdots, j_{k-1}\\ i_{k-1} = j_{k-1}}} 1  \cdot (UU^\top)_{j_{k-1},i_k}^2  \prod_{c=1}^{k-2} \langle S_{i_c}, S_{j_c}\rangle^2 \prod_{c=1}^{k-1} (UU^\top)_{j_{c-1},i_c}^2 C_{j_0} }_{B},\\
\end{align*}}
where the equality uses that $\sum_{i_k, j_{k-1}} (UU^\top)_{i_k, j_{k-1}}^2 = \|UU^\top\|_F^2 = d$.  
We first upper bound term $B$:{%
\begin{eqnarray*}
& = & \sum_{ \substack{ i_1,\cdots, i_k\\ j_0, \cdots, j_{k-1}\\ i_{k-1} = j_{k-1}} } 1  \cdot (UU^\top)_{j_{k-1},i_k}^2  \prod_{c=1}^{k-2} \langle S_{i_c}, S_{j_c}\rangle^2 \prod_{c=1}^{k-1} (UU^\top)_{j_{c-1},i_c}^2 C_{j_0}\\
 & = & \sum_{ \substack{ i_1,\cdots, i_{k-1}\\ j_0, j_1, \cdots ,j_{k-1}\\ i_{k-1} = j_{k-1} } } C_{j_0} \prod_{c=1}^{k-2} \langle S_{i_c}, S_{j_c}\rangle^2 \prod_{c=1}^{k-1} (UU^\top)_{j_{c-1},i_c}^2 \sum_{i_k} (UU^\top)_{j_{k-1},i_k}^2\\
& = & \sum_{ \substack{ i_1,\cdots, i_{k-1} \\ j_0, j_1, \cdots ,j_{k-1}\\ i_{k-1} = j_{k-1} } } C_{j_0} \prod_{c=1}^{k-2} \langle S_{i_c}, S_{j_c}\rangle^2 \prod_{c=1}^{k-1} (UU^\top)_{j_{c-1},i_c}^2  |e_{j_{k-1} } UU^\top|^2\\
& \leq & \sum_{ \substack{ i_1,\cdots, i_{k-1}\\ j_0, j_1, \cdots ,j_{k-1}\\ i_{k-1} = j_{k-1} } } C_{j_0} \prod_{c=1}^{k-2} \langle S_{i_c}, S_{j_c}\rangle^2 \prod_{c=1}^{k-1} (UU^\top)_{j_{c-1},i_c}^2  1\\
& = & \sum_{ \substack{ i_1,\cdots, i_{k-1}\\ j_0, j_1, \cdots ,j_{k-2} } } C_{j_0} \prod_{c=1}^{k-2} \langle S_{i_c}, S_{j_c}\rangle^2 \prod_{c=1}^{k-1} (UU^\top)_{j_{c-1},i_c}^2,  \\
\end{eqnarray*}
}
where the first equality is the definition of $B$, the second 
equality follows by separating out the index $i_k$, the third equality 
uses that $\sum_{i_k} (UU^\top)^2_{j_{k-1}, i_k} = \|e_{j_{k-1}} UU^\top\|_2^2$, that is, the squared
norm of the $j_{k-1}$-th row of $UU^\top$,
the inequality follows since all rows of a projection matrix $UU^\top$ have norm at most $1$,
and the final equality uses that $j_{k-1}$ no longer appears in the expression. 

We now merge our bounds for the terms $A$ and $B$ in the following way:
{%
\begin{eqnarray*}
&& \| Z\|_F^2 \\
&\leq & A+B \\
&\leq & O(d(\log n)/r) \sum_{ \substack{ i_1,\cdots,i_{k-1} \\ j_0,\cdots,j_{k-2} } } \prod_{c=1}^{k-2} \langle S_{i_c}, S_{j_c}\rangle^2 \prod_{c=1}^{k-1} (UU^\top)_{j_{c-1},i_c}^2 C_{j_0} \\
& + & \sum_{ \substack{ i_1,\cdots,i_{k-1}\\ j_0, j_1,\cdots,j_{k-2}} } C_{j_0} \prod_{c=1}^{k-2} \langle S_{i_c}, S_{j_c}\rangle^2 \prod_{c=1}^{k-1} (UU^\top)_{j_{c-1},i_c}^2 \\
& = & \left(  O(d(\log n)/r) + 1 \right) \sum_{\substack{ i_1,\cdots, i_{k-1}\\ j_0, \cdots, j_{k-2}} } \prod_{c=1}^{k-2} \langle S_{i_c}, S_{j_c}\rangle^2 \prod_{c=1}^{k-1} (UU^\top)_{j_{c-1},i_c}^2 C_{j_0}\\
& \leq & \cdots \\
& \leq &  \left(  O(d(\log n)/r) + 1 \right)^2 \sum_{ \substack{i_1,\cdots, i_{k-2} \\ j_0, \cdots,j_{k-3} }} \prod_{c=1}^{k-3} \langle S_{i_c}, S_{j_c}\rangle^2 \prod_{c=1}^{k-2} (UU^\top)_{j_{c-1},i_c}^2 C_{j_0} \\ 
& \leq & \cdots \\
& \leq & \left(  O(d(\log n)/r) + 1 \right)^{k-1} \sum_{i_1, j_0 } \prod_{c=1}^{1} (UU^\top)_{j_{c-1},i_c}^2 C_{j_0}\\
& \leq &  \left(  O(d(\log n)/r) + 1 \right)^{k-1} (d\|b\|_2^2 (\log^2 n)/r^2 + \| b\|_2^2 (\log n)/r),
\end{eqnarray*}
}

where the first two inequalities and first equality 
are by definition of $A$ and $B$ above. The first inequality follows by induction, 
since at this point we have replaced $k$ with $k-1$, and can repeat the argument,
incurring another multiplicative factor of $O(d (\log n)/r) + 1$. Repeating the
induction in this way we arrive at the last inequality. Finally, the last inequality follows by plugging in the definition of $C_{j_0}$, using
that $\sum_{i_1, j_0} (UU^\top)^2_{j_0, i_1} = d$, and { 
$$\sum_{j_0, i_1} (UU^\top)^2_{j_0, i_1} a_{j_0}^2 
= \sum_{j_0} a_{j_0}^2 \sum_{i_1} (UU^\top)^2_{j_0, i_1}
= \sum_{j_0} a_{j_0}^2 \|e_{j_0}UU^\top\|_2^2
\leq 1,$$}
where the inequality follows since each row of $UU^\top$ has norm at most
$1$, and $a$ is a unit vector.  The final result is that { 
\[
 \| Z\|_F^2  \leq  \left(  O(d(\log n)/r) + 1 \right)^{k-1} (d\|b\|_2^2 (\log^2 n)/r^2 + \| b\|_2^2 (\log n)/r).
\]}




\end{document}